%% file: main.tex
\documentclass[conference]{IEEEtran}
\IEEEoverridecommandlockouts
\usepackage[noadjust]{cite}
\usepackage{amsmath,amssymb,amsfonts,amsthm}
\usepackage{stmaryrd}
\usepackage{algorithmic}
\usepackage{graphicx}
\usepackage{textcomp}
\usepackage{xcolor}
\usepackage{hyperref}
\usepackage{cleveref}
\usepackage{tikz}
\usepackage{amsfonts,bm}

\allowdisplaybreaks

\def\BibTeX{{\rm B\kern-.05em{\sc i\kern-.025em b}\kern-.08em
    T\kern-.1667em\lower.7ex\hbox{E}\kern-.125emX}}

\newtheorem{theorem}{Theorem}[section] % Numbering theorems according to their section.
\newtheorem{proposition}[theorem]{Proposition}

\newtheorem{lemma}[theorem]{Lemma}

\theoremstyle{definition}
\newtheorem{definition}[theorem]{Definition}

\theoremstyle{remark}
\newtheorem{remark}[theorem]{Remark}

\newtheorem{example}[theorem]{Example}
\allowdisplaybreaks{}

\input{macros}
\input{games.tikzstyles}

\begin{document}

\title{The Algebra of Parity Games
%\thanks{Identify applicable funding agency here.}
}

\author{\IEEEauthorblockN{Robin Piedeleu}\IEEEauthorblockA{\textit{University College London, UK}}}

\maketitle

\begin{abstract}
In recent work, Watanabe, Eberhart, Asada, and Hasuo have shown that parity games can be seen  as string diagrams, that is, as the morphisms of a symmetric monoidal category, an algebraic structure with two different operations of composition. Furthermore, they have shown that the winning regions associated to a given game can be computed functorially, \emph{i.e.} compositionally. Building on their results, this paper focuses on the equational properties of parity games, giving them a sound and complete axiomatisation. The payoff is that any parity game can be solved using equational reasoning directly at the level of the string diagram that represents it. Finally, we translate the diagrammatic language of parity games to an equally expressive symbolic calculus with fixpoints, and equip it with its own equational theory.
\end{abstract}

\section{Introduction}
\label{sec:intro}

Parity games are the computational engine that powers a number of important problems in logic, formal verification, and synthesis. For example, solving parity games is equivalent to the model-checking problem for the modal $\mu$-calculus~\cite{bradfield2018mu}. This makes parity games fundamental to verifying the correctness of complex systems. Similarly, problems like validity and satisfiability in several modal logics can be reduced to parity game solving~\cite{wilke2001alternating}. 

Famously, the complexity-theoretic status of parity games remains open: it is not known whether they can be solved in polynomial time, though we know the problem belongs to the intersection of the UP and co-UP classes~\cite{jurdzinski1998deciding}. This continues to drive extensive research into more efficient algorithms~\cite{jurdzinski2008deterministic,czerwinski2019universal}.

Recently, Watanabe, Eberhart, Asada, and Hasuo have introduced \emph{open} parity games, allowing for a compositional treatment of these games~\cite{compositionalPG}. In their work, the labelled graphs associated to open parity games are recast as string diagrams, that is, as the morphisms of a symmetric monoidal category, an algebraic structure with two different operations of composition, sequential and parallel. They generalise the notion of winning regions to open parity games and show that these define a monoidal functor into a suitable semantic category. In other words, they demonstrate that solving (open) parity games can be done compositionally: the solutions of some game can be computed by composing the semantics of its components. This pioneering work opens up the study of parity games using the tools of algebra, in particular equational reasoning, which is the focus of our paper.

\paragraph{Contributions} While in their work, Watanabe et al define a diagrammatic syntax for open parity games, their presentation involves an infinite number of generators, one for each possible node occurring in a parity game. For example, there is one generator for a parity $5$ node owned by Player $1$ with two inputs and three outputs, one for a similar node with three inputs and one outputs, and so on. The first contribution of this work is to identify a \emph{finite} set of generators from which all open parity games can be constructed, \emph{i.e.} a set which is \emph{universal} for open parity games (Theorem~\ref{thm:universality}). 

Our second and more noteworthy contribution is to identify all the algebraic laws that govern the relationship between open parity games: more specifically, we give a finite set of equalities which is \emph{sound and complete} for the chosen semantics (Theorems~\ref{thm:soundness} and~\ref{thm:completeness}). In other words, any two open parity games with the same semantics are provably equal using our equational theory. As a result, it is possible to solve a given parity game using only equational reasoning, going one significant step further than the work of Watanabe et al. In their own words~\cite[Remark 3.16]{compositionalPG}
\begin{quote}
``an equational characterization of the equivalence of open parity games [...] seems challenging, however, given the complexity of solving parity games, and we leave it as future work.''
\end{quote}
This is precisely the problem that this paper solves.

Our final contribution is a translation of the diagrammatic presentation of parity games into a standard algebraic syntax with fixpoints. The language and its translation follow the work of Hasegawa~\cite{hasegawa2012models}. We believe it will be helpful to those who are less familiar with string diagrams and demonstrates the practical value of our work in potential implementations.

The reader might legitimately wonder why we would want an equational axiomatisation when equivalence of parity games (or simply solving them) is decidable. We believe our equational characterisation has independent value, for several reasons.
First and foremost, equational reasoning is compositional: the solution of a complex game can be found from iteratively solving its component games. This is particularly valuable in cases where games have some regularity that allows us to reuse previously computed solutions~\cite{rathke2014compositional,watanabe2024pareto}. Relatedly, our algebraic way of reasoning about parity games grounds the solution to each game in \emph{local} reasoning principles (\emph{i.e.} axioms). This can be especially valuable in scenarios where these solutions must be communicated, checked, or automated. 
Finally, our diagrammatic axiomatisation links to a growing body of work on related graphical formalisms, in particular to previous work on automata or Boolean satisfiability, where similar algebraic structure appears~\cite{piedeleu2023finite,gu2023complete}.

\paragraph{Outline} 
We start with few preliminaries on parity games, before introducing \emph{open} parity games and their semantics in Section~\ref{sec:prelims}. 
Then, in Section~\ref{sec:syntax-semantics}, we introduce the diagrammatic syntax we will use throughout, give its interpretation, and show that it is sufficiently expressive to encode all open parity games (Theorem~\ref{thm:universality}). In Section~\ref{sec:axiomatisation}, we give the equational theory of open parity games and prove it is sound and complete for the intended semantics (Theorems~~\ref{thm:soundness} and~\ref{thm:completeness}). In Section~\ref{sec:symbolic}, we consider an equally expressive symbolic syntax. In conclusion, we examine related work and open with a few directions for future research.

\paragraph*{Notation} We will write $\ord{n}$ for the ordinal $\{1,\dots,n\}$, $[f,g]\from X+Y\to Z$ for the co-pairing of two maps $f\from X\to Z$ and $g\from Y\to Z$, and $\entry{x}, \exit{y}$ for the two inclusions of $x\in X$ and $y\in Y$ in the coproduct $X+Y$. We will also write $!_\emptyset$ for the only map out of the empty set. We denote by $R^\star$ the transitive closure of a binary relation $R$, and by $V^*$ the set of words over some set $V$.

\section{Background on (open) parity games}
\label{sec:prelims}

We fix a natural number $\maxparity \ge 2$ and let $\Priorities = \{0,1,\dots,\maxparity\}$ be our set of \emph{priorities}. All the results in this paper are parametric in this choice of upper bound, which can always be increased in practice to accommodate parity games with arbitrarily large priorities.

\subsection{Parity games}
\label{sec:parity-games}

\begin{definition}\label{def:parity-game}
A \emph{parity game} $\PG = (V,V_0,E,\Labelling)$ is a tuple consisting of a set of \emph{positions} $V$, partitioned into two subsets $V_0$ and $V_1:=V\setminus V_0$ indicating which position belong to Player $0$ or $1$ respectively, a set of (directed) \emph{edges} $E\subseteq V\times V$, and a labelling function $\Labelling\from V\to \Priorities$, assigning a \emph{priority} to each position of the game.
\end{definition}

A \emph{play} is a sequence of positions in $V$ such that $(v_n,v_{n+1})\in E$.  A finite play is \emph{winning} for Player $0$ (resp. $1$) if it ends on a position in $V_1$ (resp. $V_0$). Winning conditions for infinite plays are more interesting and follow what is known as the \emph{parity condition}: an infinite play $(v_n)_{n\in\N}$ is \emph{winning} for Player $0$ (resp. $1$), if the maximum priority that occurs infinitely often is even (resp. odd). In other words, Player $0$ wins $p$ if $\limsup_{n\in\N} \Labelling(v_n)$ is even, otherwise it is \emph{losing} (\emph{i.e.} winning for Player $1$). A strategy for Player $0$ is a function $\sigma\from V^*V_0\to V$ where $(v_k,\sigma(v_1\dots v_k))\in E$, that is, it is a function that chooses a new position adjacent to $v_k$ for each finite sequence of positions $v_1\dots v_k$ ending in a position in $V_0$ (that Player $0$ controls). Strategies for Player $1$ are defined in the same way, replacing $0$ with $1$ in the preceding sentence.

Solving a parity game involves determining, for a given position $v$, which player has a winning strategy for any opponent strategy, starting from $v$. The key result about parity games is that they are \emph{positionally determined}: each position is winning for exactly one of the two players~\cite{emerson1991tree}.

Typically, parity games are represented as a directed graph whose diamond nodes represent positions owned by Player $0$, square nodes represent positions owned by Player $1$, with labels indicating the  corresponding priorities.
\begin{example}
\label{ex:pg}
The game $(\{v_{tl},v_{tr},v_{bl},v_{br}\},\{v_{tl},v_{bl}\}, E,\Labelling)$, where $E=\{(v_{tl},v_{tr}),(v_{tr},v_{tl}),(v_{tl},v_{bl}),(v_{bl},v_{br}),(v_{br},v_{bl})$, $(v_{br},v_{tl}),(v_{br},v_{tr})\}$
and $\Labelling : v_{tl}\mapsto 4,v_{tr}\mapsto 1,v_{bl}\mapsto 3, v_{br}\mapsto 2$, can be depicted as
$$

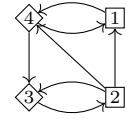
\InputIfFileExists{example-game-1.tikz}{}{\input{./tikz/example-game-1.tikz}}

$$
In this game, the top nodes $\{v_{tl},v_{tr}\}$ are winning positions for Player $0$ and the bottom ones are winning positions for Player $1$. Indeed, consider the bottom right node $v_{br}$ for example: Player $1$ can choose to go left to $v_{bl}$, at which point Player $0$ has no choice but to come back to the starting position; this defines an infinite play whose highest priority is odd (\emph{viz.} $3$), so Player $1$ wins. If Player $1$ start by going up diagonally to $v_{tl}$ node, Player $0$ would win by responding with $v_{tr}$ node. 
\end{example}

\subsection{Open parity games and their algebraic operations}
\label{sec:category-opg}

We now extend the usual definition of parity games to include incoming and outgoing edges, in order to be able to compose them~\cite[Section 2]{compositionalPG}.
\begin{definition}
\label{def:open-parity-game}
An \emph{open parity game} is a tuple $\oPG{A} = (\object{m},\object{n},V,V_0,E,\Labelling)$ where
\begin{itemize}
\item  $\object{m} =(\lin{m},\lout{m})$ and $\object{n} = (\rout{n},\rin{n})$ are pairs of natural numbers, and we call $\entries{\oPG{A}} := \ord{\lin{m}+\rin{n}}$ the set of \emph{entry positions} and $\exits{\oPG{A}} = \ord{\lout{m}+\rout{n}}$ the set of \emph{exit positions};
\item $V$ is a finite set of \emph{internal positions};
\item $V_0$ indicates the set of internal positions which belong to Player $0$ (and $V_1:= V\setminus V_0$ those that belong to Player $1$);
\item $E$ is a relation $E\subseteq (\entries{\oPG{A}}+V)\times (V+\exits{\oPG{A}})$ whose elements are called \emph{edges};
\item there is a unique edge out of each entry position and a unique edge into each exit position;
\item $\Labelling$ is a function $V\to \Priorities$ assigning a priority to each internal position.
\end{itemize}
Finally, we write $\oPG{A} \from \object{m}\to \object{n}$, and call $\object{m}$ its \emph{domain} and $\object{n}$ its \emph{codomain}.
\end{definition}
%Note that inputs and outputs of open parity games belong to Player $0$. This is merely a convention, since they have a single successor and therefore the player that owns them cannot choose where to play next. Moreover, the priority of incoming or outgoing edges is set to $0$ by convention to simplify the definition of composition below.
We represent open parity games much like plain parity games, but with entry positions as incoming edges without any source and exit positions as outgoing edges without any target (with the total order on each side read from top to bottom). 
\begin{example}\label{ex:opg}
The following 
$$

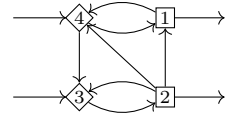
\InputIfFileExists{example-opg-1.tikz}{}{\input{./tikz/example-opg-1.tikz}}

$$
represents the open parity game $((2,0),(2,0)$, $\{v_{tl},v_{tr},v_{bl},v_{br}\}$, $\{v_{tl},v_{bl}\}$, $E,\Labelling)$ where $E$ and $\Labelling$ are defined as in Example~\ref{ex:pg}.
\end{example}
\begin{definition}
\label{def:composition}
Let $\oPG{A}\from\object{m} \to \object{n}$ and $\oPG{B}\from \object{n}\to \object{o}$ be two open parity games. Their \emph{composition} $\oPG{A}\puis \oPG{B}\from \object{m}\to \object{o}$ is the open parity game $(\object{m},\object{o},V^\oPG{A}+V^\oPG{B},V_0^\oPG{A}+V_0^\oPG{B},E^{\oPG{A}\puis\oPG{B}},[\Labelling^\oPG{A},\Labelling^\oPG{B}])$, such that $E^{\oPG{A}\puis\oPG{B}} :=$
$$
E^\oPG{A} \setminus \big((\entries{\oPG{A}}+V^\oPG{A})\times \exits{\oPG{A}}\big) + E^\oPG{B} \setminus \big(\entries{\oPG{B}}\times (V^\oPG{B}+\exits{\oPG{B}})\big) + E'
$$ 
where 
\begin{align*}
E' := \big\{(v,w)\mid \exists i\,\big((v,\inr(i))\in E^\oPG{A}\land (\inl(i),w)\in E^\oPG{B}\big)\big\}\cup\quad\\
\; \big\{(w,v)\mid \exists i\,\big((w,\inr(\rout{o}+i))\in E^\oPG{B}\land (\inl(\lin{m}+i),v)\in E^\oPG{A}\big)\big\}
\end{align*}
Note that we abuse notation slightly: we write $E^\oPG{A} \setminus \big((\entries{\oPG{A}}+V^\oPG{A})\times \exits{\oPG{A}}\big)$ to mean $E^\oPG{A} \setminus \big((\entries{\oPG{A}}+V^\oPG{A})\times \inr^*(\exits{\oPG{A}})\big)$, where $\inr^*(\exits{\oPG{A}})$ is the image of the right insertion $\inr$ in the coproduct $V^\oPG{A}+\exits{\oPG{A}}$.
\end{definition}
Intuitively, the composition of two open parity games identifies the right nodes (entries and exits) of the first morphism with the left nodes (entries and exits) of the second and hides them, \emph{i.e.}, replaces them each with an edge in the appropriate direction. 
\begin{example}\label{ex:composition}
The open parity game from Example~\ref{ex:opg} is the composite of the following two games:
$$

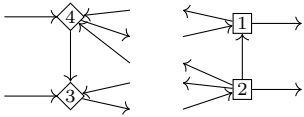
\InputIfFileExists{example-opg-composition.tikz}{}{\input{./tikz/example-opg-composition.tikz}}

$$
\end{example}
The example below will be the unit for this operation of composition.
\begin{example}\label{ex:id-game}
We call the game $(\object{m},\object{m},\emptyset,\emptyset,\{(i,i)\mid i\in \ord{\lin{m}+\rout{n}} \},!_\emptyset)$ the \emph{identity} over $\object{m}$ and write it as $\id_{\object{m}}\from \object{m} \to \object{m}$.
\end{example}
\begin{definition}
\label{def:monoidal-product}
Let $\oPG{A}_1\from \object{m}_1 \to \object{n}_1$ and $\oPG{A}_2\from \object{m}_2\to \object{n}_2$ be two open parity games. Their \emph{monoidal product} $\oPG{A}_1\otimes \oPG{A}_2\from \object{m}_1+\object{m}_2\to \object{n}_1+\object{n}_2$ is the open parity game $((\lin{m}_1+\lin{m}_2, \lout{m}_1+\lout{m}_2), (\rout{n}_1+\rout{n}_2,\rin{n}_1+\rin{n}_2), V^{\oPG{A}_1}+V^{\oPG{A}_2},V_0^{\oPG{A}_1}+V_0^{\oPG{A}_2},E^{\oPG{A}_1\otimes\oPG{A}_2},[\Labelling^{\oPG{A}_1},\Labelling^{\oPG{A}_2}])$, with
$$
E^{\oPG{A}_1\otimes\oPG{A}_2} := E^{\oPG{A}_1} + E^{\oPG{A}_2}_\downarrow
$$
where $E^{\oPG{A}_2}_\downarrow$ shifts entry and exit positions of $E^{\oPG{A}_2}$ by $\lin{m}_1+ \rin{n}_1$ and $\lout{m}_1+ \rout{n}_1$ respectively: more precisely, it is the relation defined by $(v,v')\in E^{\oPG{A}_2}_\downarrow$ iff
\begin{itemize}
\item $v,v'\in V^{\oPG{A}_2}$ and $(v,v')\in E^{\oPG{A}_2}$;
\item $v = \lin{m}_1+ \rin{n}_1 + i$ for $i\in \entries{\oPG{B}}$ and $(i,v')\in E^{\oPG{A}_2}$;
\item $v' = \lout{m}_1+ \rout{n}_1 + o$ for $o\in \exits{\oPG{B}}$ and $(v,o)\in E^{\oPG{A}_2}$.
\end{itemize}
\end{definition}
Intuitively, the monoidal product of two games simply juxtaposes the two graphs and renumbers the entry and exit positions of the second game. The first example below will be the unit for this operation; the second one is another simple open parity game that reorder the entry and exit positions. 
\begin{example}
\label{def:empty-swap}
(\emph{i}) The \emph{empty game} $(0,0)\to (0,0)$ is the open parity game given by $((0,0),(0,0),\emptyset,\emptyset,\emptyset,!_\emptyset)$.

(\emph{ii}) The \emph{swap} game $\sigma^{\object{m}}_{\object{n}} = (\object{m}+\object{n},\object{n}+\object{m},\emptyset,\emptyset,E^\sigma,!_\emptyset)$ where the sum of pairs represents the pair of component-wise sums and $E^\sigma := \big\{(k,\rout{n}+k) \mid k\in\ord{\lin{m}} \big\}\cup \big\{(\lin{m}+k, k) \mid k\in\ord{\rout{n}} \big\} \cup \big\{(k,\lout{m}+k) \mid k\in\ord{\rin{n}} \big\}\cup \big\{(\rin{n}+k, k) \mid k\in\ord{\lout{m}} \big\}$. For instance,
$$
\sigma^{(2,0)}_{(1,0)} = 
\InputIfFileExists{swap-2-1.tikz}{}{\input{./tikz/swap-2-1.tikz}}

$$
\end{example}

%With the composition, monoidal product and swap game, open parity games already form a symmetric monoidal category~\cite{compositionalPG}.
 
%We now equip open parity games with a \emph{trace}, a form of feedback with which we can connect exit positions to entry positions and hide them.
%\begin{definition}
%\label{def:trace}
%Given an open parity game $\oPG{A}\from m+\ell \to n+\ell$, let $\Tr^{m,n}_\ell\oPG{A}\from m\to n$ be the open parity game given by $(m,n,V^\oPG{A},V_0^\oPG{A},E^{m,n}_\ell,\Labelling^\oPG{A})$ where $E^{m,n}_\ell$ is the relation defined by $(v,v')\in E^{m,n}_\ell$ iff $(v,v')\in E^\oPG{A}$ or there exists $i,o\in\ord{\ell}$ such that $(m+i,n+o)\in \big({E^\oPG{A}\cap (\ord{\ell}\times \ord{\ell})}\big)^\star$ and $(v,i)\in E^\oPG{A}$ and $(o,v')\in E^\oPG{A}$.
%\end{definition}
%In other words, the edges of the trace of an open parity game $\oPG{A}\from m+\ell \to n+\ell$ correspond to paths starting from a position in $\ord{m}+V^\oPG{A}$ and ending in a position in $\ord{n}+V^\oPG{A}$, potentially looping for several iterations via one of the edges of $E^\oPG{A}$ which connect exit positions in $\ord{\ell}$ to entry positions in $\ord{\ell}$.

Up to a suitable notion of syntactic equivalence, it is possible to show that open parity games form a %traced 
symmetric monoidal category~\cite[Theorem 2.10]{compositionalPG}. However, since we will only care about a coarser notion of semantic equivalence, we postpone the definition of the category of open parity game to the next subsection. 
%We have defined a traced monoidal category of open parity games~\cite[Theorem 2.10]{compositionalPG}.
%\begin{theorem}
%\label{thm:opg-traced-monoidal-category}
%Open parity games form a traced symmetric monoidal category with the operations of composition, identities, monoidal product, swap, and trace as defined above. 
%\end{theorem}

\subsection{Equivalence of open parity games}
\label{sec:semantic-opg}

The main question associated with a given parity game is that of determining the winning positions for each player. We now extend this notion in a compositional way to open parity games, following the work of Watanabe et al. If parity games are positionally determined, the entry and exit positions of open parity games require us to keep track of more information. Plays in an open parity game are defined analogously to plays in their closed counterpart: they are (finite or infinite) sequences of positions that follow edges of the game. Note that, for open parity games, plays that include some exit position $o$ are necessarily finite and end with $o$ (this is because we have defined exit positions as sink nodes in the graph corresponding to the game). Conversely, a play (finite or infinite) that contains some entry position $i$ necessarily starts with $i$ (because entry positions are defined to be source nodes in the corresponding graph). 

Winning conditions for infinite plays in open parity games are the same as for standard parity games: an infinite play is \emph{winning} for Player $0$ (resp. $1$), if it satisfies the parity condition, \emph{i.e.}, if the maximum priority that occurs infinitely often is even (resp. odd). Similarly, finite plays that end at an internal position of an open parity game are winning for the player who does not own the last position. The situation is less clear for finite play that ends at an exit position. How to interpret the outcome of such a play? To answer this question, Watanabe et al move from a purely qualitative (winning or losing) definition to a quantitative one: the denotation of a finite play that ends at an exit position $o$ is the pair $(o,k)$where $k$ is the highest priority encountered along the play.
\begin{definition}
\label{def:play-denotation}
Given a play $(v_i)_{i\in I}$ on an open parity game, its \emph{denotation} $\playsem{(v_i)_{i\in I}}_\oPG{A}$ is given by
		$(o,k)$ if $I$ is finite, ends on exit position $o$ and $k=\max\{\Labelling(v_i) \mid i\in I\}$;
		by $\top$ if $I$ is infinite and $(v_i)_{i\in I}$ is winning for Player $0$; 
		by $\bot$ if $I$ is infinite and $(v_i)_{i\in I}$ is winning for Player $1$.
\end{definition} 
\begin{definition}
\label{def:strategy}
In an open parity game $\oPG{A}\from \object{m} \to \object{n}$, a strategy for Player $0$, or $0$-\emph{strategy} is a function $\sigma\from (\entries{\oPG{A}}+V)^*V_0\to V+ \exits{\oPG{A}}$ where $(v_k,\sigma(v_1\dots v_k))\in E$. In other words, it is a function that chooses a new position adjacent to $v_k$ for each finite sequence of positions $v_1\dots v_k$ ending in a position that Player $0$ controls. Strategies for Player $1$ are defined in an analogous way. 

For $0$-strategy $\sigma_0$ and $1$-strategy $\sigma_1$, \emph{the play induced by} $\sigma_0$ and $\sigma_1$, starting at position $v$, is the sequence $\Play_{\sigma_0,\sigma_1}^v$ of positions given by $v_1 = v$ and $v_{n+1} = \sigma_0(v_1\dots v_n)$ if $v_n\in V_0$ and $v_{n+1} = \sigma_1(v_1\dots v_n)$ otherwise. 
\end{definition}
\begin{definition}
\label{def:denotation-semantics}
The \emph{denotation} $\possem{i,\sigma_0}_\oPG{A}$ of an entry position $i\in\entries{\oPG{A}}$ and $0$-strategy $\sigma_0$ in an open parity game $\oPG{A}\from \object{m} \to \object{n}$ is $\textbf{lose}$ if there is a $1$-strategy $\sigma_1$ such that $\playsem{\Play_{\sigma_0,\sigma_1}^v}_\oPG{A} =\bot$; otherwise $\possem{i,\sigma_0}_\oPG{A}$ is the set 
$$
\bigg\{\playsem{\Play_{\sigma_0,\sigma_1}^v}_\oPG{A} \,\bigg|\, \sigma_1 \text{ is a $1$-strategy}, \playsem{\Play_{\sigma_0,\sigma_1}^v}_\oPG{A} \neq \top\bigg\}
$$

The \emph{semantics} of $\oPG{A}\from \object{m} \to \object{n}$, written $\oPGsem{\oPG{A}}$, is the collection:
\begin{equation*}
\begin{gathered}
\big\{(T,i)\in \Powerset{\exits{\oPG{A}}\times\Priorities}\times \entries{\oPG{A}} \qquad\qquad\qquad\qquad\qquad\qquad\qquad\\
\; \,\big|\, \possem{i,\sigma_0}_\oPG{A} \subseteq T, \possem{i,\sigma_0}_\oPG{A} \neq \textbf{lose} \text{ for some $0$-strategy } \sigma_0\big\}
\end{gathered}
\end{equation*}
\end{definition}
Intuitively, $(T,i)$ belongs to $\oPGsem{\oPG{A}}$, if Player $0$ has a strategy such that, for any strategy of Player $1$, the induced play starting at $i$ will end in \emph{some} exit positions $o$ of $T$, with maximum priority encountered $k$ if $(o,k)\in T$. So, in a sense, the smaller $T$ is, the more control Player $0$ has over the game, and the less influence Player $1$'s strategy has. In particular, if $(\emptyset,i)\in \oPGsem{\oPG{A}}$ then Player $0$ has a winning strategy starting from $i$. Conversely, if $\oPGsem{\oPG{A}}$ contains no pair $(T,i)$, then Player $1$ has a winning strategy starting from $i$. In all intermediate cases, the outcome of the game is not determined in the usual sense, as it will depends on what happens when the play exits $\oPG{A}$. Crucially, for any open parity game if $\oPG{A}$ has not exit positions, that is, if its codomain is $(0,0)$, then $\oPGsem{\oPG{A}}$ contains precisely the information about which player wins from which entry position. In this sense, the semantics of open parity games extends the usual notion of winning region of parity games in a faithful way.
\begin{definition}\label{def:opg-equivalence}
We define the following equivalence relation on open parity games: $\oPG{A}\oPGeq \oPG{B}$ if the two games $\oPG{A},\oPG{B}\from m \to n$ have the same semantics; we say that $\oPG{A}$ and $\oPG{B}$ are \emph{equivalent}.
\end{definition}
\begin{remark}
\label{rmk:upward-closed}
It is possible to organise the relations that provide the semantics of open parity games into a traced monoidal category~\cite[Section 5]{compositionalPG} whose morphisms are certain \emph{monotone} relations, defined as the coKleisli category of a composite comonad based on the work of Grellois and Melli{\`e}s~\cite{grellois2015finitary}. This makes $\Winning$ into a traced symmetric monoidal functor, called the \emph{winning position functor}.
%Note that the semantic functor in that work is defined in a different way than our semantics~\cite[Definition 5.5]{compositionalPG}. For the authors, the semantics of an open parity game $\oPG{A}\from m \to n$ is the set of all pairs in $\Priorities\times\ord{n}$ \emph{that contain} some denotations of plays induced by strategies of Player $0$ and $1$:
%$$\Winning(\oPG{A}) = \{(T,i) \mid \playsem{i}_\oPG{A}\subseteq T, i\in\ord{m}\}$$
%As a result the semantics of any given entry position is an \emph{upward-closed} subset of $\Priorities\times\ord{n}$. This is required in that work, because the target semantics is a category of \emph{monotone} relations. 

%However, for practical purposes, we have found computing with these relations less intuitive than reasoning with strategies, and have decided not to use this machinery here. Instead, we have followed Watanabe's subsequent work~\cite[Definitions 9 and 10]{watanabe2024pareto}, in which the author only keeps track of the minimal elements of the upward-closed set associated to each entry position, as we do in Definition~\ref{def:denotation-semantics} above. Since any upward-closed sets is fully characterised by its minimal elements, no information about the game is lost from this perspective. And, using this idea, we will be able to organise open parity games \emph{up to equivalence} into a traced symmetric monoidal category.

However, we have found computing the composition of these relations to be complicated in practice and often less intuitive than reasoning with strategies directly, at least for the purposes of this paper. For this reason, we only use the semantics to define equivalence of open parity games and use the results of Watanabe et al to show that they form a trace symmetric monoidal category \emph{up to equivalence}, thereby warranting our use of string diagrams in the following section.

There is another difference between our work and that of Watanabe et al: in this paper, we do not use the \emph{Int-construction}, a well-known technique to turn a traced monoidal category into a compact-closed one~\cite{joyal1996traced}. Given our aims, we have found it unnecessary here: from the perspective of open parity games as graphs, we do not need to define a trace operation explicitly, since the graphs we consider admit incoming and outgoing edges both in their domain and codomain, and composition is therefore sufficient to form loops. From the semantic perspective, we can rely on the results of Watanabe et al to show that open games up to equivalence already form a compact-closed category and therefore we have no need for the explicit use of the Int-construction (though we use it implicitly by our appeal to results in that paper).
\end{remark}
%While we do not need our semantics to define a functor, 
%We need $\oPGeq$ to interact well with the algebraic operations defined earlier.
%\begin{proposition}\label{thm:congruence}
%The equivalence relation $\oPGeq$ is a congruence for the operations of composition, monoidal product and trace.
%\end{proposition}
%\begin{proof}
%This is a consequence of the fact that the \emph{winning position functor} in Watanabe et al is a traced symmetric monoidal functor~\cite[Theorem 5.8]{compositionalPG}. 
%\end{proof}
%The result below warrants our use of string diagrams in the following section.
\begin{theorem}
\label{thm:opg-traced-monoidal-category}
Open parity games up to $\oPGeq$, with the operations of composition (with identities as unit), monoidal product (with the empty game as unit), swap 
%and trace 
defined above, form a 
%traced 
symmetric monoidal category, which we call $\OPGcat$.
\end{theorem}
\begin{proof}
The first remark is that the morphisms of $\OPGcat$ with the operations of composition and monoidal product defined earlier are the same as those of the symmetric monoidal category of the same name in the work of Watanabe et al~\cite[Definition 2.11]{compositionalPG}. As we have mentioned, the main difference is that they define it using the Int-construction, but it suffices to unpack their definition to realise that our operations on open parity games and theirs are the same (because we build them straight-away using bidirectional composition, while they proceed in two steps, starting from a traced monoidal category of games whose entry positions are all on the left and whose exit position are all on the right, and apply the Int-construction to it).

Secondly, we will need the fact that the semantic mapping $\Winning$ is a symmetric monoidal functor~\cite[Theorem 5.8]{compositionalPG} into a certain symmetric monoidal category of monotone relations~\cite[Section 4]{compositionalPG}, which the authors call $\IntC\left(\textbf{FinScottL}_{!_\maxparity}^{op}\right)$, though we do not need to define it here.
Now, consider four open parity games such that $\oPG{A}\oPGeq\oPG{A}'$ and $\oPG{B}\oPGeq\oPG{B}'$, that is, such that $\Winning(\oPG{A})=\Winning(\oPG{A}')$ and $\Winning(\oPG{B})=\Winning(\oPG{B}')$. The functoriality of $\Winning$ implies that $\Winning(\oPG{A}\puis\oPG{B}) = \Winning(\oPG{A})\bm{;}\Winning(\oPG{B}) = \Winning(\oPG{A}')\puis\Winning(\oPG{B}') =  \Winning(\oPG{A}'\puis\oPG{B}')$, where the middle operation $\bm{;}$ 	is the composition in $\IntC\left(\textbf{FinScottL}_{!_\maxparity}^{op}\right)$. From this, we can conclude that $\oPG{A}\puis\oPG{B}\oPGeq \oPG{A}'\puis\oPG{B}'$ or, in other words, that $\oPGeq$ is a congruence for the operation of composition. Moreover, if $\Winning\big((\oPG{A}\puis \oPG{B})\puis \oPG{C}\big) = \big(\Winning(\oPG{A})\bm{;}\Winning(\oPG{B})\big)\bm{;}\Winning(\oPG{C}) = \Winning(\oPG{A})\bm{;}\big(\Winning(\oPG{B})\bm{;}\Winning(\oPG{C})\big) = \Winning\big(\oPG{A}\puis (\oPG{B}\puis \oPG{C})\big)$, by associativity of composition in $\IntC\left(\textbf{FinScottL}_{!_\maxparity}^{op}\right)$, so that $(\oPG{A}\puis \oPG{B})\puis \oPG{C}\oPGeq \oPG{A}\puis (\oPG{B}\puis \oPG{C})$, showing that the composition of open parity games up to equivalence is associative.

A similar reasoning, using the fact that $\Winning$ is 
%traced 
symmetric monoidal shows that $\oPGeq$ is a congruence for the monoidal product, 
%and the trace, 
and that the laws of symmetric monoidal categories hold up to equivalence.
\end{proof}

\section{A diagrammatic calculus for parity games}
\label{sec:syntax-semantics}

\subsection{Diagrammatic syntax}
\label{sec:diagrammatic-syntax}

Our syntax will be defined as a \emph{prop}, a strict symmetric monoidal category whose objects are all products of some finite set of objects. More specifically, it will be a prop which is freely generated by a given signature, a set of generating objects and morphisms. We will represent the morphisms of this category as \emph{string diagrams}, the natural two-dimensional syntax of symmetric monoidal categories. For a gentle introduction to string diagrams we refer the reader to Selinger's classic survey~\cite{selinger2011survey} or Piedeleu and Zanasi's recent text~\cite{piedeleu2023introduction}. The reader who would like a formal account of the construction of free props from a given signature will find it in the work of Baez, Coya and Rebro on network theory~\cite[Appendix B]{baez2018props}. 
\begin{definition}
Let $\Syntax$ be the two-coloured prop freely generated by the following objects and morphisms:
\begin{itemize}
\item the two generating objects, $\objr$ (right) and $\objl$ (left), whose identities we will depict respectively as the directed wires $\idright$ and $\idleft$;
\item the generating morphisms in~\eqref{eq:syntax-acyclic} and~\eqref{eq:syntax-loops} below. 
\begin{equation}
  \label{eq:syntax-acyclic}
\splitzero,\losezero,\splitone,\loseone\,,\join,\start, \priorityedge{p} 
\end{equation}
\begin{equation}
\label{eq:syntax-loops}

\InputIfFileExists{cap-down.tikz}{}{\input{./tikz/cap-down.tikz}}
,
\InputIfFileExists{cup-down.tikz}{}{\input{./tikz/cup-down.tikz}}

\end{equation}
\end{itemize}
\end{definition}
Thus, objects of $\Syntax$ are words over the alphabet $\{\objr,\objl\}$ and the monoidal product is given by concatenation on objects, with the empty word $\epsilon$ denoting the unit. Morphisms of $\Syntax$ are vertical and horizontal compositions of the generators above, potentially including wire crossings (\emph{e.g.} $
\InputIfFileExists{sym.tikz}{}{\input{./tikz/sym.tikz}}
$) and identity wires ($\idright$ and $\idleft$) \emph{up to} the laws of symmetric monoidal categories (Fig.~\ref{fig:smc-axioms}). The direction of the arrows on the wires denotes their \emph{type}: for example, $\join$ represents an operation of type $\objr\objr\to\objr$, while $
\InputIfFileExists{cap-down.tikz}{}{\input{./tikz/cap-down.tikz}}
$ has type $\objl\objr\to \varepsilon$. As for open parity games, we call $v$ the \emph{domain} of a given diagram $d\from v\to w$ its \emph{codomain}. We depict a generic diagram $d\from v\to w$ as a box with $|v|$ ordered wires labelled by the elements of $v$ on the left and $|w|$ wires labelled by the elements of $w$ on the right. When we have $n$ parallel of the same type, say $\objr$, we depict them as a single directed wire labelled by a natural number label, as $\idright^{\!\!\!\!\!\! n}$ . We can compose these diagrams in two different ways: horizontally by connecting the right wires of the first diagram to the left wires of the second (when the types match), and vertically by simply juxtaposing two diagrams: 
$$

\InputIfFileExists{seq-compose.tikz}{}{\input{./tikz/seq-compose.tikz}}
\qquad \quad 
\InputIfFileExists{par-compose.tikz}{}{\input{./tikz/par-compose.tikz}}

$$ 
The symmetry is drawn as a wire crossing $
\InputIfFileExists{sym.tikz}{}{\input{./tikz/sym.tikz}}
$, and the unit for the monoidal product as the empty diagram $
\InputIfFileExists{empty-diag.tikz}{}{\input{./tikz/empty-diag.tikz}}
$. With this representation the laws of Fig.~\ref{fig:smc-axioms} become diagrammatic near-tautologies. We call \emph{inputs} the incoming wires of a diagram, and \emph{outputs} its outgoing wires. More formally, the inputs (resp. outputs) of $d\from v\to w$ are the set of positions of the word $v$ which are $\objr$ (resp. $\objl$) and the position of $w$ which are $\objl$ (resp. $\objr$). Note that we will often omit the arrows on wires when they are unambiguous, to avoid overloading diagrams. 

We will introduce the semantics of the generators in terms of open parity games in the next subsection, but let us give some intuition now: $\splitzero$ represents a single Player $0$ position with one entry and two exit positions, $\losezero$ represents a single Player $0$ position connected to one entry and no exit positions (which is therefore loosing); $\splitone$ and $\loseone$ represent the corresponding games for Player $1$; $\join$ and $\start$ represent two open parity games with a single exit and two or zero entry position respectively (it does not matter to whom they belong since there is not choice of exit); $\priorityedge{k}$ represents a position with priority $k$, connected to a single entry and exit position; finally, $
\InputIfFileExists{cap-down.tikz}{}{\input{./tikz/cap-down.tikz}}
$ and $
\InputIfFileExists{cup-down.tikz}{}{\input{./tikz/cup-down.tikz}}
$ are pieces of syntax that allow us to bend wires and form feedback loops.

Finally, we will sometimes use generalised $\join$ and $\start$ on $n$-wires, depicted as $
\InputIfFileExists{mergen.tikz}{}{\input{./tikz/mergen.tikz}}
$ and $
\InputIfFileExists{startn.tikz}{}{\input{./tikz/startn.tikz}}
$. These are definable from $\join$ and $\start$ using the operations of composition, monoidal product, and wire crossings.

\begin{figure*}
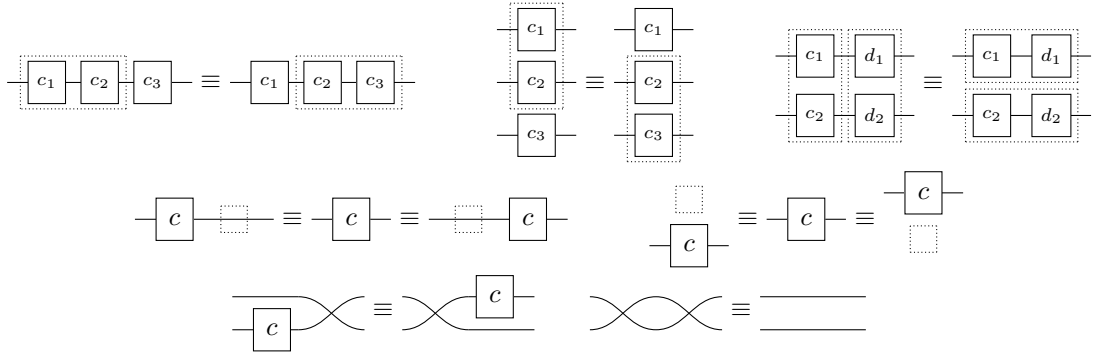

\begin{equation*}
\begin{gathered}
{
\InputIfFileExists{smc/sequential-associativity.tikz}{}{\input{./tikz/smc/sequential-associativity.tikz}}
 \SMCeq 
\InputIfFileExists{smc/sequential-associativity-1.tikz}{}{\input{./tikz/smc/sequential-associativity-1.tikz}}
} \qquad \quad {
\InputIfFileExists{smc/parallel-associativity.tikz}{}{\input{./tikz/smc/parallel-associativity.tikz}}
 \SMCeq 
\InputIfFileExists{smc/parallel-associativity-1.tikz}{}{\input{./tikz/smc/parallel-associativity-1.tikz}}
}\qquad\quad  {
\InputIfFileExists{smc/interchange-law.tikz}{}{\input{./tikz/smc/interchange-law.tikz}}
\SMCeq
\InputIfFileExists{smc/interchange-law-1.tikz}{}{\input{./tikz/smc/interchange-law-1.tikz}}
 }
 \\
{
\InputIfFileExists{smc/unit-right.tikz}{}{\input{./tikz/smc/unit-right.tikz}}
 \SMCeq \diagbox{c}{}{} \SMCeq 
\InputIfFileExists{smc/unit-left.tikz}{}{\input{./tikz/smc/unit-left.tikz}}
}
\qquad\quad
{ 
\InputIfFileExists{smc/parallel-unit-above.tikz}{}{\input{./tikz/smc/parallel-unit-above.tikz}}
 \SMCeq \diagbox{c}{}{} \SMCeq  
\InputIfFileExists{smc/parallel-unit-below.tikz}{}{\input{./tikz/smc/parallel-unit-below.tikz}}
}
\\
{
\InputIfFileExists{smc/sym-natural.tikz}{}{\input{./tikz/smc/sym-natural.tikz}}
 \SMCeq 
\InputIfFileExists{smc/sym-natural-1.tikz}{}{\input{./tikz/smc/sym-natural-1.tikz}}
}
\qquad
{
\InputIfFileExists{smc/sym-iso.tikz}{}{\input{./tikz/smc/sym-iso.tikz}}
 \SMCeq 
\InputIfFileExists{id2.tikz}{}{\input{./tikz/id2.tikz}}
}
\end{gathered}
\end{equation*}
\caption{Axioms of symmetric monoidal categories. The variables $c,c_1,c_2,\dots, d_1,\dots $ range over all possible diagrams of the appropriate type (labels of domains and codomains are omitted for legibility). The dotted frames on the first line indicate the order of application of composition and the monoidal product; on the second line they indicate identities when they surround a plain wire, or the unit for the monoidal product when they surround an empty space.}
\label{fig:smc-axioms}
\end{figure*}

\subsection{Semantics of diagrams}
\label{sec:semantics-int-construction}

We are now ready to define the semantics of our diagram as a symmetric monoidal functor from $\Syntax$ to $\Semantics$. Since $\Syntax$ is free, we only need to specify how to interpret each of the generating objects and morphisms.
\begin{definition}\label{def:semantic-functor}
Let $\sem{\cdot}\from \Syntax\to \Semantics$ be the symmetric monoidal functor fully specified by
\begin{itemize}
\item $\sem{\objr} = (1,0)$ and $\sem{\objl} = (0,1)$ on objects;
\item the mapping given in Fig.~\ref{fig:interpretation} on morphisms.
\end{itemize}
\end{definition}
Then, by symmetric monoidal functoriality, the semantics of the (horizontal) composite of two diagrams $c\from u\to v$ and $d\from v\to w$ is given by $\sem{c};\sem{d}$, following Definition~\ref{def:composition}, the semantics of the monoidal product (vertical juxtaposition) of two diagrams $d_1\from v_1\to w_1$ and $d_2\from v_2\to w_2$ is given by $\sem{d_1}\otimes\sem{d_2}$, following Definition~\ref{def:monoidal-product}, and the semantics of wire crossings is $\sem{
\InputIfFileExists{sym-vxw.tikz}{}{\input{./tikz/sym-vxw.tikz}}
} = \sigma^{\sem{v}}_{\sem{w}}$, the swap game from Example~\ref{def:empty-swap} (\emph{ii}). 
\begin{figure*}[t]
\begin{align*}
\sem{\splitzero} & = \big(1,2,\{v\},\{v\},\{(\entry{1},v),(v,\exit{1}),(v,\exit{2})\},\{v\mapsto 0\}\big) &
\\
\sem{\losezero\,} & = \big(1,0,\{v\},\{v\},\{(\entry{1},v)\},\{v\mapsto 0\}\big) &
\\
\sem{\splitone} & = \big(1,2,\{v\},\emptyset,\{(\entry{1},v),(v,\exit{1}),(v,\exit{2})\},\{v\mapsto 0\}\big) & \sem{
\InputIfFileExists{cap-down.tikz}{}{\input{./tikz/cap-down.tikz}}
}  = \big(1,1,\emptyset,\emptyset,\{(\entry{1},\exit{1})\},!_\emptyset\big)
\\
\sem{\loseone\,} & = \big(1,0,\{v\},\emptyset,\{(\entry{1},v)\},\{v\mapsto 0\}\}\big) & \sem{
\InputIfFileExists{cup-down.tikz}{}{\input{./tikz/cup-down.tikz}}
}  = \big(1,1,\emptyset,\emptyset,\{(\entry{1},\exit{1})\},!_\emptyset\big)
\\
\sem{\priorityedge{p}} & = \big(1,1,\{v\},\emptyset,\{(\entry{1},v), (v,\exit{1})\},\{v\mapsto p\}\big) &
\\
\sem{\join} & = \big(2,1,\{v\},\emptyset,\{(\entry{1},v), (\entry{2},v),(v,\exit{1})\},\{v\mapsto 0\}\big) &
\\
\sem{\,\start} & = \big(0,1,\{v\},\emptyset,\{(v,\exit{1})\},\{v\mapsto 0\}\big) &
\end{align*}
\caption{Interpretation of generators as open parity games. Recall that the first two components are the type (domain and codomain) of the corresponding game, the third is the set of internal positions (either empty for the cups and caps, or some arbitrary single node labelled $v$ for all the other generators), the fourth is the set of nodes that belong to Player $0$, the fifth is the set of edges, and the last one is the labelling function. Notice that the cup and cap have the same denotation}
\label{fig:interpretation}
\end{figure*}

%An immediate consequence of our choice of semantics is that forming loops in the syntax corresponds to taking the trace of the corresponding open parity game. More precisely, for any diagram $d\from \objr^{\ell+m}\to \objr^{\ell+m}$, 
%$$\sem{\tikzfig{trace-d-m-n}} \oPGeq \Tr^\ell_{m,n}\sem{d}$$

Now that we have fixed the syntax and semantics of our diagrammatic language, we can check that it is sufficiently expressive to encode arbitrary open parity games. Before reading the proof of the following theorem, we encourage the reader to consult the example in Fig.~\ref{fig:ex-opg-to-diag}.
\begin{figure*}
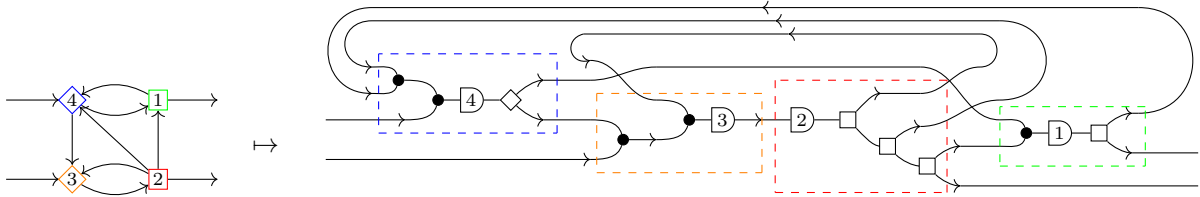

$$

\InputIfFileExists{example-opg-colour.tikz}{}{\input{./tikz/example-opg-colour.tikz}}
 \quad\mapsto\quad 
\InputIfFileExists{example-game-diag-1.tikz}{}{\input{./tikz/example-game-diag-1.tikz}}

$$
\caption{Translation of the open parity game in Example~\ref{ex:opg} into a diagram with the same semantics. The framed sub-diagrams encode the positions of the same colour on the left.}\label{fig:ex-opg-to-diag}
\end{figure*}
\begin{theorem}[Universality]
\label{thm:universality}
For any open parity game $\oPG{A}\from m\to n$ there exists a diagram $d_\oPG{A}\from \objr^m\to\objr^n$ such that $\sem{d_\oPG{A}} \oPGeq  \oPG{A}$.
\end{theorem}
\begin{proof}
To prove this, we can encode the graph of $\oPG{A}$ directly, as follows: first, for every Player $0$ position $v\in V^\oPG{A}$ with priority $k$, $m_v$ incoming edges and $n_v$ outgoing edges let $d_v$ be the diagram obtained by sequentially composing $m_v-1$ $\join$-generators, followed by a $\priorityedge{k}$, and $n_v-1$ $\splitzero$ generators, as depicted below:
$$
d_v := 
\InputIfFileExists{0-position-encoding.tikz}{}{\input{./tikz/0-position-encoding.tikz}}

$$
Similarly for Player $1$ positions, let
$$
d_v := 
\InputIfFileExists{1-position-encoding.tikz}{}{\input{./tikz/1-position-encoding.tikz}}

$$
Then, we can simply juxtapose all $d_v$ for $v\in V^\oPG{A}$ vertically using the monoidal product, obtaining a single diagram with $\sum_{v\in V^\oPG{A}} m_v + m$ inputs and $\sum_{v\in V^\oPG{A}} n_v + n$ outputs. We can then use $
\InputIfFileExists{cap-down.tikz}{}{\input{./tikz/cap-down.tikz}}
$ and $
\InputIfFileExists{cup-down.tikz}{}{\input{./tikz/cup-down.tikz}}
$ to connect each output that are is no exit positions of $\oPG{A}$ to the corresponding input that is not an entry position, forming the edges of $\oPG{A}$. These must match exactly, by the handshaking lemma for directed graphs, since $\sum_{v\in V^\oPG{A}} m_v = \sum_{v\in V^\oPG{A}}$ is the number of internal edges of $\oPG{A}$.

The diagram we obtain is $d_\oPG{A}$, which satisfies $\sem{d_\oPG{A}} \oPGeq  \oPG{A}$ by construction.
\end{proof}
Note that there is another way to prove the previous theorem: from our definition of the semantics of open parity games, it is clear that every game is equivalent to a finite one. We can therefore always construct $d_\oPG{A}$ to be a in a simple form that we will later call the \emph{normal form} of diagrams (Definition~\ref{def:acyclic-normal-form}).

\section{Axiomatisation}
\label{sec:axiomatisation}

This section contains our main technical results: an equational theory which we show is sound and complete for equivalence of open parity games. In other words, we will give a finite set of axioms which correspond to valid semantic equalities (soundness) and which are sufficient to derive any semantic equality between diagrams (completeness).

\subsection{Equational theory}
\label{sec:equational-theory}

The axioms of Parity Game Algebra (PGA) are given in Fig.~\ref{fig:compact-closed-axioms} to~\ref{fig:parity-axiom}. Let us explain them in more details.
\begin{enumerate}
\item Fig.~\ref{fig:compact-closed-axioms} contains the axioms of \emph{compact-closed categories}~\cite{kelly1980coherence}, a common feature of graphical calculi. These axioms allow us to bend and straighten wires at will, while only keeping track of their direction. Moreover, we include in this group an axiom which allows us to remove a single loop with no entry or exit points. From the perspective of parity games, such diagrams represent inaccessible games, whose outcome is therefore irrelevant.
\item Fig.~\ref{fig:axioms-distributive-lattices} contains the axioms of \emph{distributive lattices}. We can think of $\splitzero$ as a form of disjunction and $\splitone$ as a form of conjunction. Some of the axioms admit intuitive readings in terms of strategies. For example, $\losezero$ is a counit for $\splitzero$, witnessing the fact that Player $0$ would never choose to lose the game immediately, and would always continue on the branch that is not blocked by $\losezero$ instead. The same holds for Player $1$ with $\splitone$ and $\loseone$.
\item Fig.~\ref{fig:axioms-priorities} contains the axioms for the algebraic theory of priorities (also known as the colouring comonad~\cite{grellois2015finitary}). They guarantee that priorities distribute over $\join$ and $\start$, and that they compose sequentially by taking their maximum, a key feature of how they interact in parity games. Finally, since $0$ is the minimum of $\Priorities$, it acts as a unit for composition, \emph{i.e.} the identity.
\item Fig.~\ref{fig:axioms-priorities-lattices} contains axioms for the distributive law of priorities over lattices: in other words, the priority labels distribute over all four lattice operations (including constants). These equalities are self-evident from the point of view of strategies. The semantics of open games can be formulated explicitly as a the coKleisli category of a certain comonad obtained by composing two comonads via a distributive law~\cite{compositionalPG}. These axioms give an equational presentation of this law, though we will not make the connection with distributive laws precise in this paper.
\item Fig.~\ref{fig:parity-axiom} contains the axiom that encodes the parity condition. It states that, when the loop priority $p$ is even, Player $1$ has only one non-losing strategy: choosing the first exit and avoiding the loop which would be winning for Player $0$. In this case, Player $0$ will have two strategies: passing control to Player $1$ or taking the second exit. On the other hand, when $p$ is odd, Player $0$ has only one non-losing strategy: choosing the second exit, since Player $1$ has a winning strategy by taking the loop.  
\end{enumerate}
\begin{figure}
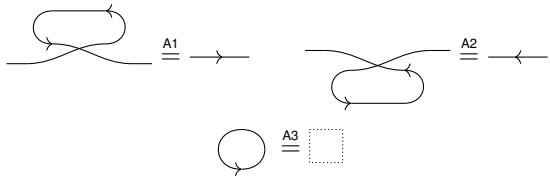

\[

\InputIfFileExists{yanking.tikz}{}{\input{./tikz/yanking.tikz}}
 \labeleq{A1} 
\InputIfFileExists{id.tikz}{}{\input{./tikz/id.tikz}}
 \qquad 
\InputIfFileExists{yanking-op.tikz}{}{\input{./tikz/yanking-op.tikz}}
 \labeleq{A2} 
\InputIfFileExists{id-op.tikz}{}{\input{./tikz/id-op.tikz}}

\]
\[

\InputIfFileExists{loop.tikz}{}{\input{./tikz/loop.tikz}}
 \labeleq{A3} 
\InputIfFileExists{empty-diag.tikz}{}{\input{./tikz/empty-diag.tikz}}

\]
\caption{Axioms of compact-closed categories (on the first line) and loop removal (on the second line)}
\label{fig:compact-closed-axioms}
\end{figure}
\begin{figure}
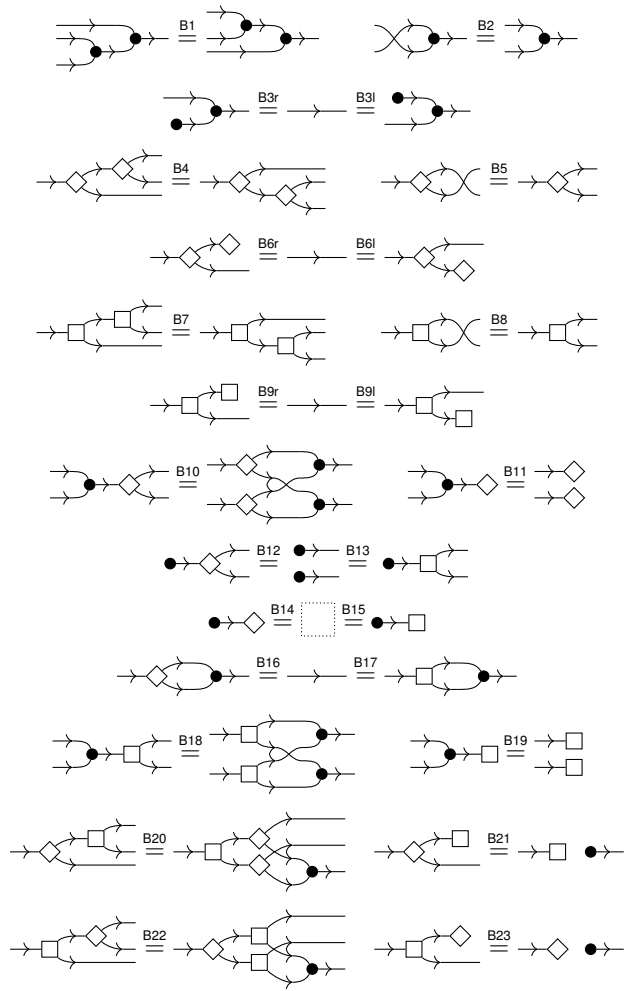

% Merge is commutative monoid
\[

\InputIfFileExists{idxmerge-merge.tikz}{}{\input{./tikz/idxmerge-merge.tikz}}
 \labeleq{B1} 
\InputIfFileExists{mergexid-merge.tikz}{}{\input{./tikz/mergexid-merge.tikz}}
 \qquad 
\InputIfFileExists{sym-merge.tikz}{}{\input{./tikz/sym-merge.tikz}}
\labeleq{B2}
\InputIfFileExists{merge.tikz}{}{\input{./tikz/merge.tikz}}

\] 
\[

\InputIfFileExists{idxstart-merge.tikz}{}{\input{./tikz/idxstart-merge.tikz}}
 \labeleq{B3r} 
\InputIfFileExists{id.tikz}{}{\input{./tikz/id.tikz}}
 \labeleq{B3l} 
\InputIfFileExists{startxid-merge.tikz}{}{\input{./tikz/startxid-merge.tikz}}

\]
% Player 0 is cocommutative comonoid
\[

\InputIfFileExists{0-0xid.tikz}{}{\input{./tikz/0-0xid.tikz}}
 \labeleq{B4} 
\InputIfFileExists{0-idx0.tikz}{}{\input{./tikz/0-idx0.tikz}}
 \qquad 
\InputIfFileExists{0-sym.tikz}{}{\input{./tikz/0-sym.tikz}}
 \labeleq{B5} 
\InputIfFileExists{player-0.tikz}{}{\input{./tikz/player-0.tikz}}

\]
\[

\InputIfFileExists{0-lose0xid.tikz}{}{\input{./tikz/0-lose0xid.tikz}}
 \labeleq{B6r} 
\InputIfFileExists{id.tikz}{}{\input{./tikz/id.tikz}}
 \labeleq{B6l} 
\InputIfFileExists{0-idxlose0.tikz}{}{\input{./tikz/0-idxlose0.tikz}}

\]
% Player 1 is cocommutative comonoid
\[

\InputIfFileExists{1-1xid.tikz}{}{\input{./tikz/1-1xid.tikz}}
 \labeleq{B7} 
\InputIfFileExists{1-idx1.tikz}{}{\input{./tikz/1-idx1.tikz}}
 \qquad 
\InputIfFileExists{1-sym.tikz}{}{\input{./tikz/1-sym.tikz}}
 \labeleq{B8} 
\InputIfFileExists{player-1.tikz}{}{\input{./tikz/player-1.tikz}}

\]
\[

\InputIfFileExists{1-lose1xid.tikz}{}{\input{./tikz/1-lose1xid.tikz}}
 \labeleq{B9r} 
\InputIfFileExists{id.tikz}{}{\input{./tikz/id.tikz}}
 \labeleq{B9l} 
\InputIfFileExists{1-idxlose1.tikz}{}{\input{./tikz/1-idxlose1.tikz}}

\]
\[

\InputIfFileExists{merge-0.tikz}{}{\input{./tikz/merge-0.tikz}}
 \labeleq{B10} 
\InputIfFileExists{0x0-sym-mergexmerge.tikz}{}{\input{./tikz/0x0-sym-mergexmerge.tikz}}

\qquad 

\InputIfFileExists{merge-lose0.tikz}{}{\input{./tikz/merge-lose0.tikz}}
 \labeleq{B11} 
\InputIfFileExists{lose0xlose0.tikz}{}{\input{./tikz/lose0xlose0.tikz}}

\]
\[

\InputIfFileExists{start-0.tikz}{}{\input{./tikz/start-0.tikz}}
 \labeleq{B12}\; 
\InputIfFileExists{startxstart.tikz}{}{\input{./tikz/startxstart.tikz}}
 \labeleq{B13}\; 
\InputIfFileExists{start-1.tikz}{}{\input{./tikz/start-1.tikz}}

\]
\[

\InputIfFileExists{start-lose0.tikz}{}{\input{./tikz/start-lose0.tikz}}
 \labeleq{B14} 
\InputIfFileExists{empty-diag.tikz}{}{\input{./tikz/empty-diag.tikz}}
 \labeleq{B15}
\InputIfFileExists{start-lose1.tikz}{}{\input{./tikz/start-lose1.tikz}}
 
\]
\[

\InputIfFileExists{0-merge.tikz}{}{\input{./tikz/0-merge.tikz}}
 \labeleq{B16} 
\InputIfFileExists{id.tikz}{}{\input{./tikz/id.tikz}}
 \labeleq{B17} 
\InputIfFileExists{1-merge.tikz}{}{\input{./tikz/1-merge.tikz}}
 
\]
\[

\InputIfFileExists{merge-1.tikz}{}{\input{./tikz/merge-1.tikz}}
\labeleq{B18}
\InputIfFileExists{1x1-sym-mergexmerge.tikz}{}{\input{./tikz/1x1-sym-mergexmerge.tikz}}

\qquad

\InputIfFileExists{merge-lose1.tikz}{}{\input{./tikz/merge-lose1.tikz}}
 \labeleq{B19} 
\InputIfFileExists{lose1xlose1.tikz}{}{\input{./tikz/lose1xlose1.tikz}}

\]
\[

\InputIfFileExists{0-1xid.tikz}{}{\input{./tikz/0-1xid.tikz}}
 \labeleq{B20} 
\InputIfFileExists{1-0x0-idxmerge.tikz}{}{\input{./tikz/1-0x0-idxmerge.tikz}}
\quad 
\InputIfFileExists{0-lose1xid.tikz}{}{\input{./tikz/0-lose1xid.tikz}}
 \labeleq{B21} 
\InputIfFileExists{lose1-start.tikz}{}{\input{./tikz/lose1-start.tikz}}

\]
\[

\InputIfFileExists{1-0xid.tikz}{}{\input{./tikz/1-0xid.tikz}}
 \labeleq{B22} 
\InputIfFileExists{0-1x1-idxmerge.tikz}{}{\input{./tikz/0-1x1-idxmerge.tikz}}
\quad
\InputIfFileExists{1-lose0xid.tikz}{}{\input{./tikz/1-lose0xid.tikz}}
 \labeleq{B23} 
\InputIfFileExists{lose0-start.tikz}{}{\input{./tikz/lose0-start.tikz}}

\]
\caption{Distributive lattice axioms}
\label{fig:axioms-distributive-lattices}
\end{figure}
\begin{figure}
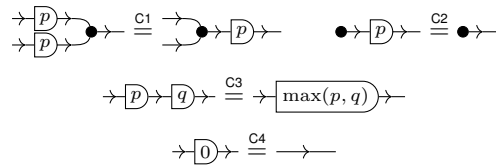

\[

\InputIfFileExists{kxk-merge.tikz}{}{\input{./tikz/kxk-merge.tikz}}
 \labeleq{C1} 
\InputIfFileExists{merge-k.tikz}{}{\input{./tikz/merge-k.tikz}}
 \qquad 
\InputIfFileExists{start-k.tikz}{}{\input{./tikz/start-k.tikz}}
 \labeleq{C2} 
\InputIfFileExists{start.tikz}{}{\input{./tikz/start.tikz}}

\]
\[

\InputIfFileExists{k-l.tikz}{}{\input{./tikz/k-l.tikz}}
 \labeleq{C3} 
\InputIfFileExists{max-kl.tikz}{}{\input{./tikz/max-kl.tikz}}
 
\]
\[
\priorityedge{0} \labeleq{C4} 
\InputIfFileExists{id.tikz}{}{\input{./tikz/id.tikz}}
 
\]
\caption{Priority axioms}
\label{fig:axioms-priorities}
\end{figure}
\begin{figure}
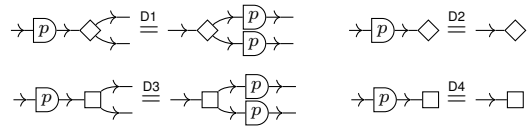

\[

\InputIfFileExists{k-0.tikz}{}{\input{./tikz/k-0.tikz}}
 \labeleq{D1} 
\InputIfFileExists{0-kxk.tikz}{}{\input{./tikz/0-kxk.tikz}}
 \qquad 
\InputIfFileExists{k-lose0.tikz}{}{\input{./tikz/k-lose0.tikz}}
 \labeleq{D2} 
\InputIfFileExists{lose-0.tikz}{}{\input{./tikz/lose-0.tikz}}

\]
\[

\InputIfFileExists{k-1.tikz}{}{\input{./tikz/k-1.tikz}}
 \labeleq{D3} 
\InputIfFileExists{1-kxk.tikz}{}{\input{./tikz/1-kxk.tikz}}
 \qquad 
\InputIfFileExists{k-lose1.tikz}{}{\input{./tikz/k-lose1.tikz}}
 \labeleq{D4} 
\InputIfFileExists{lose-1.tikz}{}{\input{./tikz/lose-1.tikz}}

\]
\caption{Axioms for the distributive law of priorities over lattices}
\label{fig:axioms-priorities-lattices}
\end{figure}
\begin{figure}
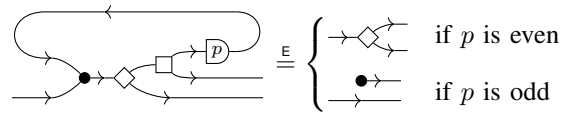

\[

\InputIfFileExists{parity-axiom.tikz}{}{\input{./tikz/parity-axiom.tikz}}
 \labeleq{E} \begin{cases}
			
\InputIfFileExists{player-0.tikz}{}{\input{./tikz/player-0.tikz}}
 & \text{if $p$ is even}\\[7pt]
            
\InputIfFileExists{startxid.tikz}{}{\input{./tikz/startxid.tikz}}
 & \text{if $p$ is odd}
		 \end{cases}
\]
\caption{Parity axiom}
\label{fig:parity-axiom}
\end{figure}

Finally, as mentioned earlier, since $\Syntax$ is defined to be a symmetric monoidal category, our diagrams also obey the laws of Fig.\ref{fig:smc-axioms}. We do not think of these axioms as part of our equational theory, but rather as a form of structural equivalence between diagrams, which we will often use implicitly. When we do use it explicitly, we will denote it by $\SMCeq$. 

We write $\eqPG$ for the smallest congruence (w.r.t to vertical and horizontal compositions) that includes the axioms of PGA. In practice, diagrammatic reasoning works like a two-dimensional generalisation of standard algebraic reasoning: if we find a sub-diagram that matches one side of an axiom in a larger diagram, we can replace it with the other side of the axiom (the left and right hand side of any axiom have the same type)~\cite[Section 2.1]{piedeleu2023introduction}

\begin{theorem}[Soundness]
\label{thm:soundness}
For any two diagrams $c,d\from v\to w$, if  $c\eqPG d$ then $\sem{c} \oPGeq \sem{d}$.
\end{theorem}
\begin{proof}
It suffices to verify that all the axioms of PGA correspond to valid semantic equivalences. We have already given an intuitive account of some of the axioms in the description above; let us cover one example in more detail:
$$
d_{lhs} := 
\InputIfFileExists{0-1xid.tikz}{}{\input{./tikz/0-1xid.tikz}}
 \labeleq{B20} 
\InputIfFileExists{1-0x0-idxmerge.tikz}{}{\input{./tikz/1-0x0-idxmerge.tikz}}
 =: d_{rhs}
$$
Both diagrams have a single entry position, so their semantics is the set of sets (of denotations of) possible plays starting from $1$ that oppose a Player $0$ strategy to a Player $1$ strategy (Definition~\ref{def:denotation-semantics}). Both diagrams correspond to games that only allow finite plays, so the denotation of any play will be a pair of an exit position and a priority (which in this case, will always be zero, since all priorities are zero here). 

For the lhs, Player $1$ has two strategies: move from the only internal position they control to either of the first two exit positions. Player $0$ also has two strategies: move to the internal position controlled by Player $1$ or to the last exit position. Playing these strategies against each other and taking the denotation of the induced plays gives:
$$
\oPGsem{\sem{d_{lhs}}} = \big\{(T,1)\mid \{(1,0), (2,0)\}\subseteq T \text{ or } \{(3,0)\}\subseteq T\big\}
$$
For the rhs, Player $0$ has four strategies: two exit choices for each of the two internal positions that they control. Player $1$ on the other hand only has two strategies: move to either of the two nodes that Player $0$ controls. Putting these together, we also have 
\begin{align*}
\oPGsem{\sem{d_{lhs}}} = \big\{(T,1) \,\mid\; &\{(1,0), (2,0)\}\subseteq T & 
\\
\; \text{ or } &\{(1,0), (3,0)\}\subseteq T&
\\
\; \text{ or } &\{(2,0), (3,0)\}\subseteq T &
\\
\;\text{ or } &\{(3,0)\}\subseteq T & \big\}
\end{align*}
which is equal to $\oPGsem{\sem{d_{lhs}}}$.
The soundness of all other axioms can be verified analogously.
\end{proof}
\begin{figure*}
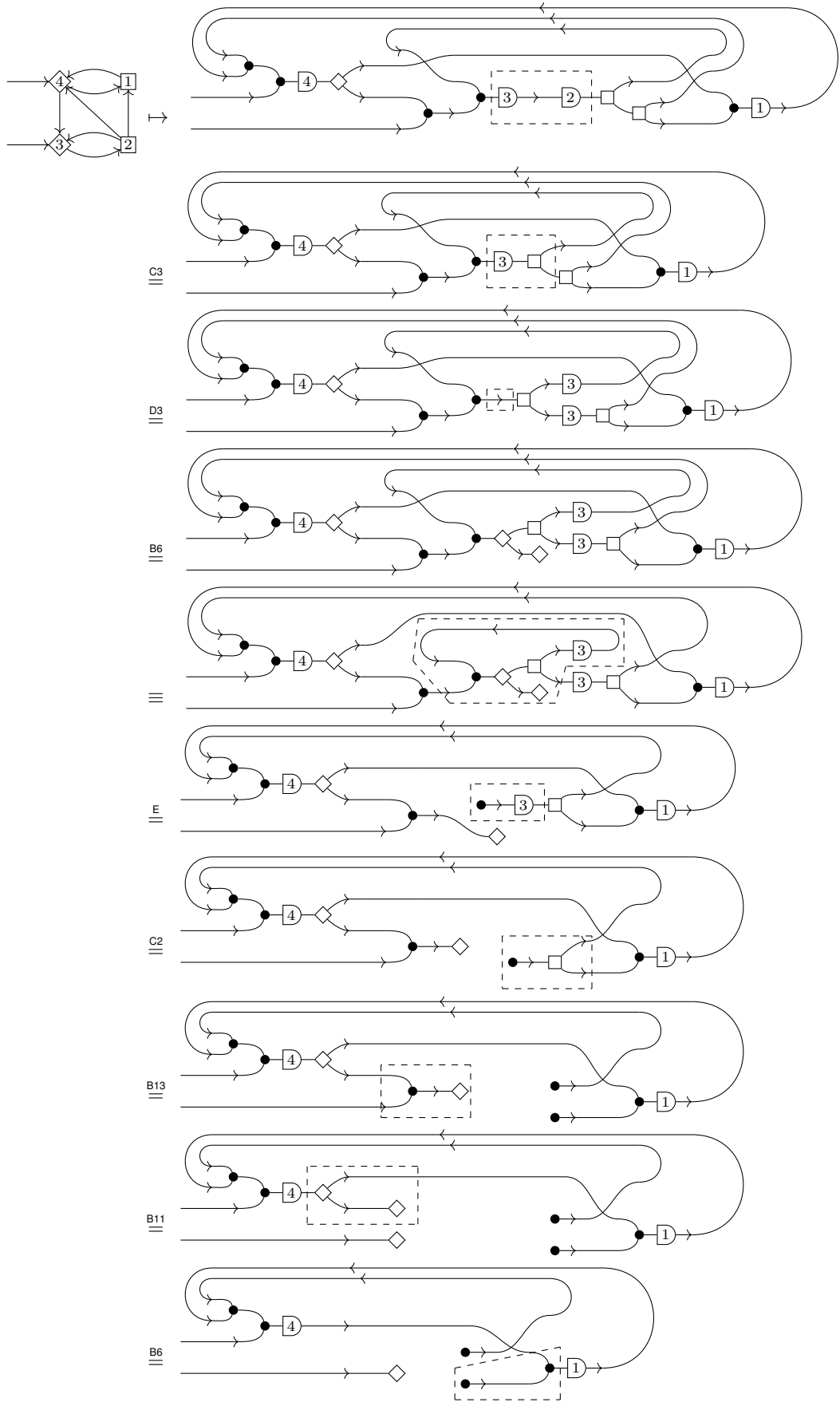

\begin{align*}

\InputIfFileExists{example-opg-1-no-outputs.tikz}{}{\input{./tikz/example-opg-1-no-outputs.tikz}}
\;&\mapsto\; 
\InputIfFileExists{example-solution.tikz}{}{\input{./tikz/example-solution.tikz}}
 \\
&\labeleq{C3}\; 
\InputIfFileExists{example-solution-1.tikz}{}{\input{./tikz/example-solution-1.tikz}}

\\  
&\labeleq{D3}\; 
\InputIfFileExists{example-solution-2.tikz}{}{\input{./tikz/example-solution-2.tikz}}

 \\
& \labeleq{B6}\; 
\InputIfFileExists{example-solution-3.tikz}{}{\input{./tikz/example-solution-3.tikz}}

 \\
& \SMCeq\; 
\InputIfFileExists{example-solution-3-bis.tikz}{}{\input{./tikz/example-solution-3-bis.tikz}}

   \\
&\labeleq{E}\; 
\InputIfFileExists{example-solution-5.tikz}{}{\input{./tikz/example-solution-5.tikz}}

   \\
&\labeleq{C2}\; 
\InputIfFileExists{example-solution-6.tikz}{}{\input{./tikz/example-solution-6.tikz}}

   \\
&\labeleq{B13}\; 
\InputIfFileExists{example-solution-7.tikz}{}{\input{./tikz/example-solution-7.tikz}}

   \\
&\labeleq{B11}\; 
\InputIfFileExists{example-solution-8.tikz}{}{\input{./tikz/example-solution-8.tikz}}

\\
& \labeleq{B6}\; 
\InputIfFileExists{example-solution-9.tikz}{}{\input{./tikz/example-solution-9.tikz}}

\end{align*}
\caption{Worked example (1/2). We solve the game via diagrammatic equational reasoning; the matching sub-diagram for each equation is highlighted in a dashed box at each step. Continued on the next page.}
\label{fig:example-solution-1}
\end{figure*}
\begin{figure*}
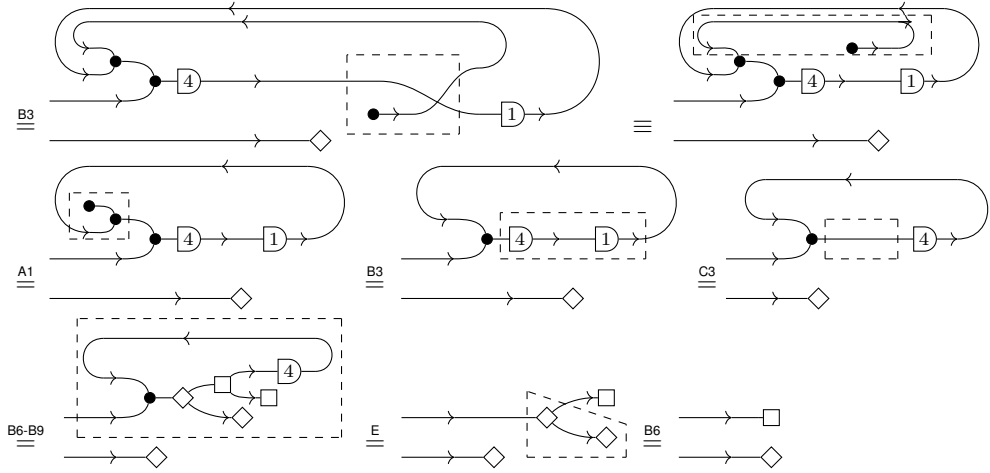

\begin{align*}
\labeleq{B3}&\; 
\InputIfFileExists{example-solution-10.tikz}{}{\input{./tikz/example-solution-10.tikz}}
\;
\SMCeq \; 
\InputIfFileExists{example-solution-10-bis.tikz}{}{\input{./tikz/example-solution-10-bis.tikz}}

\\
 \labeleq{A1} &\;
\InputIfFileExists{example-solution-11.tikz}{}{\input{./tikz/example-solution-11.tikz}}

\labeleq{B3}\; 
\InputIfFileExists{example-solution-12.tikz}{}{\input{./tikz/example-solution-12.tikz}}
 
\labeleq{C3} 
\InputIfFileExists{example-solution-14.tikz}{}{\input{./tikz/example-solution-14.tikz}}
 
\\
\labeleq{B6-B9} &\quad  
\InputIfFileExists{example-solution-15.tikz}{}{\input{./tikz/example-solution-15.tikz}}

\;\;\labeleq{E}\; 
\InputIfFileExists{example-solution-16.tikz}{}{\input{./tikz/example-solution-16.tikz}}
 \;\labeleq{B6}\; 
\InputIfFileExists{example-solution-17.tikz}{}{\input{./tikz/example-solution-17.tikz}}

\end{align*}
\caption{Worked example (2/2). The last diagram confirms the informal reasoning of Example~\ref{ex:pg}: the top left position is winning for Player $0$ and the bottom left one is winning for Player $1$.}
\label{fig:example-solution-2}
\end{figure*}

\begin{remark}[Free prop vs free traced prop]
We could have presented our syntax as a free \emph{traced} prop instead of a free prop. 
%On the one hand, this would have allowed us to avoid the use of the Int-construction to interpret diagrams directly in $\OPGcat$. On the other hand, 
However, free traced props are less standard in the literature, and the original work of Watanabe et al prefers a free prop too, perhaps because a compact closed syntax is simpler  to use and to present, axiomatically. Another reason to prefer a free prop is that the diagrams one can draw in a compact closed syntax resembles more closely the directed graphs  commonly used to represent parity games. 
%The two approaches are equivalent, of course, and the Int-construction in our case is no more than a syntactic sleight-of-hand to obtain a compact closed syntax. 
\end{remark}

\subsection{Completeness}
\label{sec:completeness}

Our proof of completeness follows a normal form argument, which proceeds in two steps:
\begin{enumerate}
\item First, we explain how to normalise diagrams corresponding to games that do not contain loops, \emph{i.e.} for which there are no non-trivial path from a position to itself. 
\item Then, we normalise arbitrary diagrams, by removing loops, thereby rewriting them to acyclic ones. 
\end{enumerate}

\subsubsection{Acyclic case}
\label{sec:acyclic-games} 
Formally, we say that a diagram $\objr^m\to\objr^n$ is \emph{acyclic} if it is composed only of generators from~\eqref{eq:syntax-acyclic}. As we said, these correspond to particularly simple open games, for which there are no infinite plays.
%and the winning conditions for these games are very simple: Player $0$ wins a play if Player $1$ owns the last node.

 The following normal form for acyclic diagrams is a generalisation of conjunctive normal form (CNF) for Boolean formulas. The difference is that it includes a layer of $\priorityedge{p}$. Here, by \emph{layer} (of a certain subset) of generators we mean a diagram composed exclusively of these generators, possibly including identities and wire-crossings.
\begin{definition}\label{def:acyclic-normal-form}
Any diagram $d\from \objr^m\to \objr^n$ is in \emph{normal form} when it factors into a layer of $\splitzero,\losezero$, followed by a layer of $\splitone,\loseone$, followed by a layer of $\priorityedge{p}$, followed by a layer of $\join,\start$.
\end{definition}
\begin{theorem}
\label{thm:acyclic-nf-unique}
If $c,d\from \objr^m\to \objr^n$ are two acyclic diagrams in normal form such that $\sem{c}\oPGeq\sem{d}$, then $c\eqPG d$.
\end{theorem}
\begin{proof}
This proof is the diagrammatic equivalent of showing that any two Boolean formulas in CNF are equal iff they contain the same clauses (\emph{i.e.} that lists of lists up to associativity, commutativity and idempotence are equivalent to sets of sets). The only difference comes from how we deal with the priority labels; but, since they distribute over other generators, the proof is similar.

First, we can always assume that $c$ and $d$ contain no sub-diagrams of the form
$

\InputIfFileExists{0-lose0xid.tikz}{}{\input{./tikz/0-lose0xid.tikz}}
$ or  $ 
\InputIfFileExists{0-idxlose0.tikz}{}{\input{./tikz/0-idxlose0.tikz}}

$
by applying co-unitality to rewrite them to an identity wire. Similarly, we can assume that the second layer of $c$ and $d$ does not contain sub-diagrams of the form 
$

\InputIfFileExists{1-lose1xid.tikz}{}{\input{./tikz/1-lose1xid.tikz}}
$ or $
\InputIfFileExists{1-idxlose1.tikz}{}{\input{./tikz/1-idxlose1.tikz}}

$
By the definition of $\sem{\cdot}$ (Fig.~\ref{fig:interpretation}), the semantics of acyclic diagrams are open games whose graph do not contain any loops and therefore only allow finite plays. By definition, the first layer of $d$ is made of $\splitzero$ or $\losezero$ generators. Thus, given some entry $i\in\ord{m}$ of $\sem{d}$, assume we first have $n_0$ sequentially composed $\splitzero$ generators (or a single $\losezero$ if $n_0=0$) connected to the $i$-th input in the first layer of $d$. Semantically, each of the possible $n_0+1$ outputs of these $n_0$ $\splitzero$ nodes, corresponds to a strategy of Player $0$, which we call $\sigma_0^k$ for $1\leq k\leq n_0+1$. 

Now, assume also that, we have $(n^1_j)_{1\leq j\leq n_0+1}$ $\splitone$ nodes (or $\loseone$ if $n^1_j=0$) connected to each of the $n_0+1$ outputs of this layer, the outputs of which are each followed by a single priority generator $\priorityedge{p}$, $p:=p_{k,j}$, for $1\leq k\leq n_0$ and $1\leq j\leq n^1_k$. Then, each of these priority generators is connected to an output $o_{k,j}\in\ord{n}$ via $\join$, which defines a possible strategy of Player $1$. 

By Definition~\ref{def:denotation-semantics}, the denotation of the pair $i,\sigma_0^k$ for $1\leq k\leq n_0+1$ is
$$
\possem{i,\sigma_0^k}_{\sem{d}} = \{(o_{k,j},p_{k,j})\mid 1\leq j\leq n^1_k\}
$$
Thus, we see from that any two acyclic diagrams $c,d\from \objr^m\to \objr^n$ in normal form are semantically equivalent if, for each input $i\in\ord{m}$, they have the same number of $\splitzero$ connected to $i$ in the first layer (or exactly one $\losezero$), followed by the same number of $\splitone$ or $\loseone$ for each of the outputs of the first layer, each followed by the same priority for each output to which they are connected via $\join$. Hence $c$ and $d$ can only differ in the following three ways:
\begin{enumerate}
\item they have a different number of $\splitone$ that encode the same pair $(o_{k,j},p_{k,j})\in X$ in the same $X\in \possem{i}_{\oPG{A}}$;
\item they have a different number of $\splitzero$ that encode the same set $\{(o_{k,j},p_{k,j})\mid 1\leq j\leq n^1_k\}$;
\item the tree of $\splitzero$ and $\splitone$ in the first two layers is different.
\end{enumerate}

For 1) we can use the fact that $\priorityedge{p}$  distributes over $\join$ (Fig.~\ref{fig:axioms-priorities}) and the fact that $\splitone$ is idempotent, \emph{i.e.} that $\splitone$ composed with $\join$ is equal to the identity wire (Fig.~\ref{fig:axioms-distributive-lattices}), as follows: 
\begin{align*}

\InputIfFileExists{1-pxp-merge.tikz}{}{\input{./tikz/1-pxp-merge.tikz}}
 \,\labeleq{C1}\,  
\InputIfFileExists{1-merge-p.tikz}{}{\input{./tikz/1-merge-p.tikz}}

 \,\labeleq{B17}\, \priorityedge{p}
\end{align*}

For 2) we can use 1) and the same fact for $\splitzero$. It is enough to show the following, for any $b\from \objr\to\objr^n$:
\begin{align*}

\InputIfFileExists{0-dxd-merge.tikz}{}{\input{./tikz/0-dxd-merge.tikz}}
 &=  
\InputIfFileExists{0-merge-d.tikz}{}{\input{./tikz/0-merge-d.tikz}}

\\
& \,\labeleq{B17}\, \lrdiag{b}{}{n}
\end{align*}
where the first equality uses the fact that acyclic diagrams can be merged, \emph{cf}. Lemma~\ref{lem:co-copy-del} in Appendix.

For 3) we can use co-associativity and co-commutativity of $\splitzero$ and $\splitone$ (in Fig.~\ref{fig:axioms-distributive-lattices}) to show that any two layers of $\splitzero$ (resp. $\splitone$) which form a tree (\emph{i.e.} are connected as undirected graphs) and contain the same number of $\splitzero$ (resp. $\splitone$) generators are equal. 
\end{proof}
\begin{theorem}[Acyclic completeness]\label{thm:acyclic-completeness}
For any two acyclic diagrams $c,d\from \objr^m\to\objr^n$, if $\sem{c} \oPGeq \sem{d}$ then $c\eqPG d$. 
\end{theorem}
\begin{proof}
We proceed using a standard normalisation argument: given an acyclic diagram $d$, we give a procedure to rewrite it to a diagram $d'$ in normal form, using only the axioms of PGA, so that $d=d'$. Since normal forms uniquely represent each class of diagrams up to semantic equivalence (Theorem~\ref{thm:acyclic-nf-unique}), any two acyclic diagrams that denote the same open parity game will thus be equal to the same normal form, and therefore equal to each other, by transitivity of equality.

The main idea of the normalisation argument is to use the distributivity of $\priorityedge{p}$ over $\splitzero$, $\losezero$, and $\splitone$, $\loseone$ in order to push all of the priority nodes to the right of these four generators, combining them using the maximum axiom when necessary. Then we can use the distributive lattice axioms to reorganise $\splitzero$, $\losezero$, and $\splitone$, $\loseone$ in the order needed. We obtain some extra $\join$, $\start$ that we can push past the $\priorityedge{p}$ using the distributivity of the latter over the former, until we are done. While this is an easy proof by structural induction, which is very close to the proof that any Boolean formula is equivalent to one in CNF, we could not find a reference, so we include it in Appendix, \emph{cf.} Lemma~\ref{lem:acyclic-normalisation}.
\end{proof}

\subsubsection{General case}
\label{sec:general-completeness}
For games that may contain cycles (and therefore infinite plays) the key idea is that all such cycles can be removed. Indeed, intuitively, this is the case in the semantics: if a cycle is a winning play for one of the players, we should be able to rewrite it to a $\loseone$ or $\losezero$. This is precisely what the parity axiom allows us to do.
\begin{theorem}[Completeness of PGA]\label{thm:completeness}
For any two diagrams $c,d\from v\to w$, if $\sem{c} \oPGeq \sem{d}$ then $c\eqPG d$.
\end{theorem}
\begin{proof}
First, without loss of generality, we can restrict ourselves to diagrams $\objr^m\to\objr^n$, since the category is compact closed, \emph{cf.} Lemma~\ref{lem:compact} in Appendix.

Now assume that we have two diagrams $c,d\from \objr^m\to\objr^n$ such that $\sem{c} \oPGeq \sem{d}$. The key idea is that if we can show that $c$ and $d$ are equal to \emph{acyclic} diagrams, the theorem follows from acyclic completeness (Theorem~\ref{thm:acyclic-completeness}). 

We are thus going to show that any diagram $ \objr^m\to\objr^n$ is equal to an acyclic one. Reasoning by induction, it is enough to show that we can remove a single loop, \emph{i.e.}, that given an acyclic diagram $b\from \objr^{1+m}\to\objr^{1+n}$, the diagram below is provably equal to an acyclic diagram:
\begin{align}
\label{eq:trace-b}
\Tr^1_{m,n}(b)\; := \; 
\InputIfFileExists{trace-b-m-n.tikz}{}{\input{./tikz/trace-b-m-n.tikz}}

\end{align}
We now reason by induction on the number of generators in $b$. For the base case, if $b$ is composed of zero generators, it consists exclusively of identities and symmetries, and thus $\Tr^1_{m,n}(b)$ can be proved equal to an acyclic diagram using only the axioms of compact closed categories (Fig.~\ref{fig:compact-closed-axioms}). This is a standard fact about the free compact closed category on some generating objects, whose proof can be found in~\cite{abramsky2005abstract}. For the inductive case, assume that, for any acyclic diagram $b$ with up to $k$ generators, $\Tr^1_{m,n}(b)$ is equal to an acyclic diagram. We now want to show that this holds for $b$ with $k+1$ generators.

Since $b$ is acyclic, it is provably equal to one in normal form, by Theorem~\ref{thm:acyclic-completeness}. This means that there exists a diagram $b'$ and some priority $p$ such that
\begin{align*}

\InputIfFileExists{trace-b-m-n.tikz}{}{\input{./tikz/trace-b-m-n.tikz}}
\; = 
\InputIfFileExists{trace-b-m-n-1.tikz}{}{\input{./tikz/trace-b-m-n-1.tikz}}

\end{align*}
(where we can use the co-unitality of $\splitone,\loseone$ or $\splitzero,\losezero$ to get a diagram in this form). 
Then, using the axioms of compact closed categories (Fig.~\ref{fig:compact-closed-axioms}), we get
\begin{equation*}

\InputIfFileExists{trace-b-m-n-1.tikz}{}{\input{./tikz/trace-b-m-n-1.tikz}}
 \labeleq{A1} 
\InputIfFileExists{trace-b-m-n-2.tikz}{}{\input{./tikz/trace-b-m-n-2.tikz}}

\end{equation*}
We can now apply the parity axiom (Fig.~\ref{fig:parity-axiom}), for which there are two cases: 
\begin{itemize}
\item if $p$ is even, then
\end{itemize} 
\begin{align*}

\InputIfFileExists{trace-b-m-n.tikz}{}{\input{./tikz/trace-b-m-n.tikz}}
  \labeleq{E} 
\InputIfFileExists{trace-b-m-n-k-odd.tikz}{}{\input{./tikz/trace-b-m-n-k-odd.tikz}}
\labeleq{A1} 
\InputIfFileExists{trace-b-m-n-k-odd-1.tikz}{}{\input{./tikz/trace-b-m-n-k-odd-1.tikz}}

\end{align*}
\begin{itemize}
\item if $p$ is odd, then
\end{itemize}
\begin{align*}

\InputIfFileExists{trace-b-m-n.tikz}{}{\input{./tikz/trace-b-m-n.tikz}}
 \labeleq{E} 
\InputIfFileExists{trace-b-m-n-k-even.tikz}{}{\input{./tikz/trace-b-m-n-k-even.tikz}}
 \labeleq{A1} 
\InputIfFileExists{trace-b-m-n-k-even-1.tikz}{}{\input{./tikz/trace-b-m-n-k-even-1.tikz}}

\end{align*}
Note that we have strictly decreased the number of generators in $b$ in both cases, so we can apply the induction hypothesis to rewrite the resulting diagram into an acyclic one.
\end{proof}

As we claimed in the introduction, the completeness of PGA allows us to solve parity games using equational reasoning. First, notice that any parity game $(V,V_0,E,\Labelling)$ can be seen as an \emph{open} parity game $\oPG{A}\from |V|\to 0$ with all positions as entry positions (and no exits); using the universality of our syntax, we can encode $\oPG{A}$ as a diagram $d_{\oPG{A}}\from |V|\to 0$. Second, any diagram in normal form with no outputs is equal to a monoidal product of $\loseone$ and $\losezero$, indicating which entry positions are winning for each player. Hence, normalising $d_{\oPG{A}}$ solves the parity game with which we started. We refer the reader to Fig.~\ref{fig:example-solution-1}-\ref{fig:example-solution-2} for a worked example.

\section{An alternative symbolic syntax}\label{sec:symbolic}

In this section, we present a syntax which is equivalent to the diagrammatic language of this paper. This kind of correspondence of calculi with fixpoint with string diagrams for traced monoidal categories is now standard and was first proposed by Hasegawa~\cite{hasegawa2012models}. We sketch it here to show that it is also a viable approach to an algebraic treatment of (open) parity games---one which might be more immediately amenable to computer implementation, even if the connection to parity games is less intuitive than the graphical presentation.

\paragraph{Syntax} We define \emph{game-expressions} inductively with the following grammar:
\[
t ::= x\mid \priority{k}{t} \mid \choosezero{t}{t}\mid \chooseone{t}{t} \mid \bot\mid \top \mid \fix{x}{t}
\]
where $x$ ranges over some countably infinite set of variables and $k\in\Priorities$. We write $n\vdash t$ for a game expression with $n\in\N$ free variables $x_1,\dots, x_n$. We keep track of the ordering of the free variables to facilitate the diagrammatic interpretation of expressions below, where they become output wires (and are therefore ordered). The typing rules for the judgment $n\vdash t$ are the usual ones for algebraic theories, with the addition of the following fixpoint rule: from $n+1\vdash t$ we can derive $n\vdash \fix{x_{n+1}}{t}$. Of course, we consider expressions up to alpha-equivalence of bound variables. 

A game expression is intended to encode an open parity game \emph{with a single entry position}. Intuitively, free variables correspond to exit positions, $\choosezero{s}{t}$ corresponds to a Player $0$ choice between positions $s$ and $t$, $\chooseone{s}{t}$ to a Player $1$ choice between positions $s$ and $t$, $\bot$ to a loss for Player $0$, $\top$ to a win for Player $0$, $\priority{k}{t}$ to a node with in-and out-degree one, and priority $k\in\Priorities$, and $\fix{x}{t}$ to a loop joining the exit position $x$ to the start position of $t$.
\begin{example}\label{ex:opg-expression}
The open parity game 
$$

\InputIfFileExists{example-opg-1-entry1.tikz}{}{\input{./tikz/example-opg-1-entry1.tikz}}

$$
corresponds to the expression with one free variable
$$
\mu y_4.\mu y_1.\priority{4}{\priority{1}{y_1\land x_1}}\lor \mu y_3.\priority{3}{\priority{2}{y_4\land \priority{1}{y_1\land x_1} \land y_3}}
$$
\end{example}

\paragraph{Semantics} We interpret game-expressions inductively, as diagrams. As mentioned above, a similar interpretation of calculi with fixpoint into traced monoidal categories can be found in the original work of Hasegawa~\cite{hasegawa2012models}. The main idea is that a game-expression $t$ with $n$ free variables is translated into a diagram $\semexpr{n\vdash t}\from \objr\to\objr^n$.
%as open parity games, using the correspondence between games and diagrams that we have established in Theorem~\ref{thm:universality}. Following the same convention, we write $d_\oPG{A}$ for the diagram encoding the open parity game $\oPG{A}$.
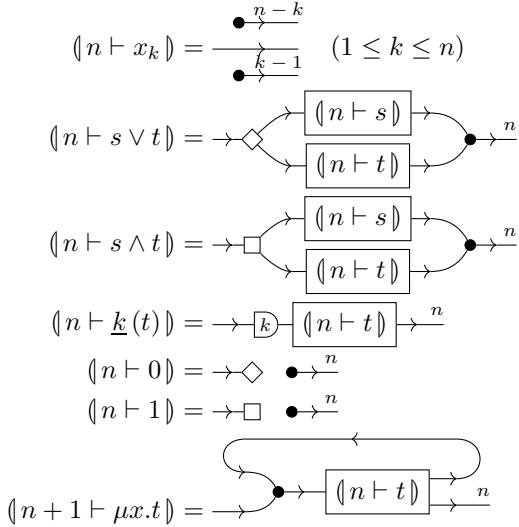
\begin{align*}
\semexpr{n\vdash x_k} &= 
\InputIfFileExists{var-translation.tikz}{}{\input{./tikz/var-translation.tikz}}
 \quad (1\leq k\leq n)
\\
\semexpr{n\vdash\choosezero{s}{t}} &= 
\InputIfFileExists{0-cxd-merge.tikz}{}{\input{./tikz/0-cxd-merge.tikz}}

\\
\semexpr{n\vdash\chooseone{s}{t}} &= 
\InputIfFileExists{1-cxd-merge.tikz}{}{\input{./tikz/1-cxd-merge.tikz}}

\\
\semexpr{n\vdash\priority{k}{t}} &= 
\InputIfFileExists{k-t.tikz}{}{\input{./tikz/k-t.tikz}}

\\
\semexpr{n\vdash 0} &= 
\InputIfFileExists{lose0-startn.tikz}{}{\input{./tikz/lose0-startn.tikz}}

\\
\semexpr{n\vdash 1} &= 
\InputIfFileExists{lose1-startn.tikz}{}{\input{./tikz/lose1-startn.tikz}}

\\
\semexpr{n+1\vdash\fix{x}{t}} &=
\InputIfFileExists{mu-translation.tikz}{}{\input{./tikz/mu-translation.tikz}}

\end{align*}
To obtain an interpretation in terms of open parity games, we simply apply  $\sem{\cdot}$ to the resulting diagrams (or, by completeness of PGA, consider two expressions semantically equal when their diagrammatic translations are equal up to $\eqPG\;$).

\paragraph{Axioms} 
With the interpretation above, up to equality of diagrams, it is easy to see that $\choosezero{}{}, \chooseone{}{}, \bot, \top$ form a distributive lattice (\textsf{B}), that $\priority{k}{\cdot}$ distributes over the lattice operations (\textsf{D}), satisfies $\priority{k'}{\priority{k}{x}} \labeleq{C3} \priority{\max{(k,k')}}{x}$, $\priority{0}{x} \labeleq{C4} x$,
and the following parity condition:
\begin{equation*}\label{eq:parity-axiom}
\fix{y}{\big(\choosezero{x_0}{(\chooseone{x_1}{\priority{k}{y}})}\big)} \labeleq{E} \begin{cases}
			\choosezero{x_0}{x_1} & \text{if $k$ is even};\\
			x_0 & \text{if $k$ is odd}.
		 \end{cases}
\end{equation*}
Since a lot of the symbolic axioms have a diagrammatic counterpart, we keep the axioms labels consistent with those of Fig.~\ref{fig:axioms-distributive-lattices}-\ref{fig:parity-axiom} where possible.

We also need axioms that allow us to swap fixpoints, $\fix{x}{\fix{y}{t}} \labeleq{F1} \fix{y}{\fix{x}{t}}$, and to remove a fixpoint $\fix{x}{t}$ when $x$ does not occur free in $t$: in this case $\fix{x}{t} \labeleq{F2} t$. Interestingly, in the presence of the other axioms, the usual fixpoint equation $\fix{x}{t} = t[\fix{x}{t}/x]$, where $t[u/x]$ denotes the substitution of $u$ for all occurrences of the variable $x$ in $t$, is derivable. This is because, as in the diagrammatic syntax, fixpoints can be removed, and any expression is equal to one in CNF, or rather, a suitable generalisation thereof that includes priorities (the symbolic counterpart of Definition~\ref{def:acyclic-normal-form}). 

Finally, using the results of this paper, or symbolic translations of the same proofs, it would not be difficult to show that these axioms are sound and complete for the intended interpretation. As before, we can use the equational theory to solve parity games. Rather than  a tedious reproduction of the diagrammatic proofs, we illustrate the idea with the following symbolic derivation. 
\begin{example} 
Coming back to Example~\ref{ex:opg-expression}, we can close the only exit position of the game by setting the free variable $x_1$ to $\top$ in the corresponding game-expression. We obtain,
\begin{align*}
& \mu y_4.\mu y_1.\priority{4}{\priority{1}{y_1\land \top}}\lor \mu y_3.\priority{3}{\priority{2}{y_4\land \priority{1}{y_1\land \top} \land y_3}}
\\
& \labeleq{B} \mu y_4.\mu y_1.\priority{4}{\priority{1}{y_1}}\lor \mu y_3.\priority{3}{\priority{2}{y_4\land \priority{1}{y_1} \land y_3}}
\\
& \labeleq{C3} \mu y_4.\mu y_1.\priority{4}{\priority{1}{y_1}}\lor \mu y_3.\priority{3}{y_4\land \priority{1}{y_1}\land y_3}
\\
& \labeleq{C3} \mu y_4.\mu y_1.\priority{4}{y_1}\lor \mu y_3.\priority{3}{y_4\land \priority{1}{y_1}\land y_3}
\\
& \labeleq{B} \mu y_4.\mu y_1.\priority{4}{y_1}\lor \mu y_3.\bot\lor \priority{3}{y_4\land \priority{1}{y_1}\land y_3}
\\
& \labeleq{D} \mu y_4.\mu y_1.\priority{4}{y_1}\lor \mu y_3.\bot\lor \big(\priority{3}{y_4\land \priority{1}{y_1}}\land \priority{3}{y_3}\big)
\\
& \labeleq{E} \mu y_4.\mu y_1.\priority{4}{y_1}\lor \bot \labeleq{B} \mu y_4.\mu y_1.\bot\lor (\top \land \priority{4}{y_1})
\\
& \labeleq{E} \mu y_4.\bot\lor \top \labeleq{B} \mu y_4.\top \labeleq{F2} \top
\end{align*}

We have thereby proved again that Player $0$ wins from the top right position of the corresponding parity game, as the last diagram in Fig.~\ref{fig:example-solution-2} had already shown.
\end{example}

\section{Discussion and Further Work}
\label{sec:conclusion}

We have presented a sound and complete calculus in which parity games can be encoded and solved using equational reasoning---this opens up promising avenues of research. 

From the complexity-theoretic perspective, our normalisation procedure requires space exponential in the size of the original diagram/game (in the same way that obtaining a CNF formula equivalent to a given Boolean formula may require exponential space). However, this is not necessarily the case on restricted classes of diagrams. It would be useful to identify sub-categories of diagrams for which solutions are tractable. Perhaps existing work on compositional tractability of parity games for some graph operations~\cite{dittmann2016graph} can point us in the right direction. 
Moreover, the local flavour of reasoning using equational axioms appears most closely related to strategy improvement algorithms for solving parity games~\cite{voge2000discrete}. Making the potential connection precise is an interesting direction of research.

%Finally, there are several existing diagrammatic languages that share some of the algebraic structure of PGA, though none of them are related to parity games. The first is the work of Piedeleu and Zanasi~\cite{piedeleu2023finite} on (finite-state) automata. Their diagrammatic language is not only compact closed, but also contains a distributive lattice on each object. Another closely related syntax is that of Gu, Piedeleu and Zanasi~\cite{gu2023complete} on Boolean satisfiability, which shares some of the same algebraic structure.

Finally, it seems possible to extend the present work to alternating parity \emph{automata}, by incorporating the additional structure found in the work of Piedeleu and Zanasi on automata~\cite{piedeleu2023finite} and adding observable actions/transitions to open parity games. Indeed, their work shares a lot of the algebraic structure found in open parity games: both diagrammatic calculi are compact closed, with a distributive lattice on each object, so combining them certainly appears possible. Semantically, in order to bridge the two pieces of work, we would require the denotation of parity automata to keep track of the actions performed by the automaton for each possible winning run (for Player $0$) starting from each entry position. Excitingly, this would also open up the way to a diagrammatic treatment of the modal $\mu$-calculus.

\begin{IEEEkeywords}
parity games, string diagrams, axiomatisation.
\end{IEEEkeywords}

\bibliographystyle{alpha}
\bibliography{refs}

\newpage

\appendix

The first result about acyclic diagrams is that we can co-copy and co-delete them at will---this gives us a diagrammatic analogue of substitution in a traditional symbolic syntax.
\begin{lemma}
For any acyclic diagram $d\from \objr^m\to\objr^n$ the following equalities are derivable:
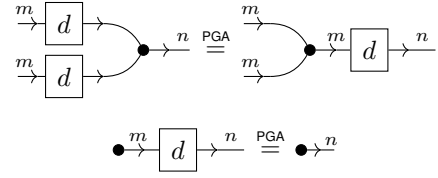
\label{lem:co-copy-del}
$$

\InputIfFileExists{dxd-merge.tikz}{}{\input{./tikz/dxd-merge.tikz}}
 \;\eqPG\; 
\InputIfFileExists{merge-d.tikz}{}{\input{./tikz/merge-d.tikz}}

$$
$$

\InputIfFileExists{start-d.tikz}{}{\input{./tikz/start-d.tikz}}
 \;\eqPG\; \start^{\!\!\! n}
$$
\end{lemma}
\begin{proof}
By structural induction. The base cases are all axioms in Fig.~\ref{fig:axioms-distributive-lattices} and Fig.~\ref{fig:axioms-priorities}. The inductive cases are as follows, where we omit object labels for readability:
\begin{itemize}
\item Composition: 
for $c$ and $d$ satisfying the induction hypothesis, we have
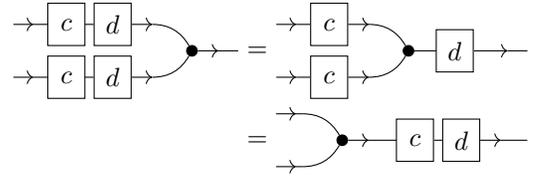
\begin{align*}

\InputIfFileExists{cdxcd-merge.tikz}{}{\input{./tikz/cdxcd-merge.tikz}}
 &= 
\InputIfFileExists{cxc-merge-d.tikz}{}{\input{./tikz/cxc-merge-d.tikz}}
\\
& = 
\InputIfFileExists{merge-cd.tikz}{}{\input{./tikz/merge-cd.tikz}}

\end{align*}
and
\begin{align*}

\InputIfFileExists{start-cd.tikz}{}{\input{./tikz/start-cd.tikz}}
 = 
\InputIfFileExists{start-d-unlabelled.tikz}{}{\input{./tikz/start-d-unlabelled.tikz}}
 = \start
\end{align*}
\item Monoidal product: 
for $c_1$ and $c_2$ satisfying the induction hypothesis, we have
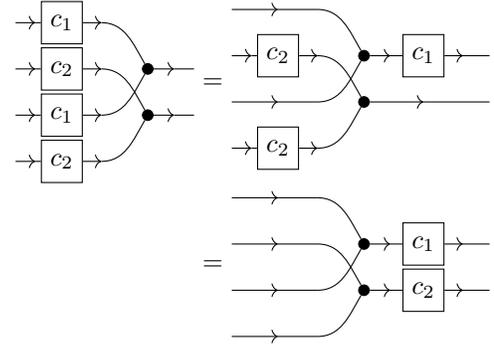
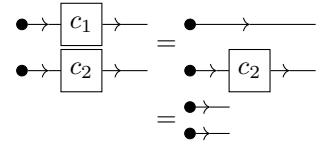
\begin{align*}

\InputIfFileExists{cxdxcxd-merge.tikz}{}{\input{./tikz/cxdxcxd-merge.tikz}}
 & = 
\InputIfFileExists{dxd-mergexmerge-cxid.tikz}{}{\input{./tikz/dxd-mergexmerge-cxid.tikz}}
\\
& = 
\InputIfFileExists{mergexmerge-cxd.tikz}{}{\input{./tikz/mergexmerge-cxd.tikz}}

\end{align*}
and
\begin{align*}

\InputIfFileExists{startxstart-cxd.tikz}{}{\input{./tikz/startxstart-cxd.tikz}}
 & = 
\InputIfFileExists{startxstart-idxd.tikz}{}{\input{./tikz/startxstart-idxd.tikz}}
\\
& = 
\InputIfFileExists{startxstart.tikz}{}{\input{./tikz/startxstart.tikz}}

\end{align*}
\end{itemize}
\end{proof}
The proof of the following lemma describes the normalisation procedure for acyclic diagrams.
\begin{lemma}
\label{lem:acyclic-normalisation}
Every acyclic diagram $d\from \objr^m\to\objr^n$ is equal to one in normal form.
\end{lemma}
\begin{proof}
We are going to give a procedure to rewrite any acyclic diagram in normal form. 
More specifically, we consider
\begin{equation}\label{eq:g-nf}

\InputIfFileExists{g-nf.tikz}{}{\input{./tikz/g-nf.tikz}}

\end{equation}
for $g\in\{\splitone,\loseone,\splitzero,\losezero,\priorityedge{k},\join,\start\}$ and $d\from\objr^{m-p}\to\objr^n$ some acyclic diagram already in normal form; we will demonstrate how to rewrite~\eqref{eq:g-nf} into a diagram in normal form, using only the axioms of PGA.
\begin{enumerate}
\item For $\priorityedge{k}$, there is some $d'$ such that
\begin{align*}

\InputIfFileExists{k-nf.tikz}{}{\input{./tikz/k-nf.tikz}}
 & = 
\InputIfFileExists{k-nf-1.tikz}{}{\input{./tikz/k-nf-1.tikz}}

\\
& \labeleq{D1}\; 
\InputIfFileExists{k-nf-2.tikz}{}{\input{./tikz/k-nf-2.tikz}}

\\
& \labeleq{D3}\; 
\InputIfFileExists{k-nf-3.tikz}{}{\input{./tikz/k-nf-3.tikz}}

\\
& \labeleq{C3}\; 
\InputIfFileExists{k-nf-4.tikz}{}{\input{./tikz/k-nf-4.tikz}}

\end{align*}
Then, we can apply a similar strategy to absorb each of the two $\priorityedge{k}$ into the diagram in the dashed box  and obtain a diagram in normal form.
\item For $\splitone$, since $d$ is in normal form, there are two cases:

\emph{i)} there is some $d'$ such that
\begin{align*}

\InputIfFileExists{1-nf.tikz}{}{\input{./tikz/1-nf.tikz}}
 & = 
\InputIfFileExists{1-nf-1.tikz}{}{\input{./tikz/1-nf-1.tikz}}

\\
& \labeleq{B8;B22}\; 
\InputIfFileExists{1-nf-3.tikz}{}{\input{./tikz/1-nf-3.tikz}}

\end{align*}
We can use Lemma~\ref{lem:co-copy-del} to rewrite the diagram in the dashed box into normal form. Then, we can apply the same rewriting strategy as above to absorb each of the two $\splitone$ into the resulting diagram and obtain a diagram in normal form.

\emph{ii)} there is some $d'$ such that
\begin{align*}

\InputIfFileExists{1-nf.tikz}{}{\input{./tikz/1-nf.tikz}}
 & = 
\InputIfFileExists{1-nf-4.tikz}{}{\input{./tikz/1-nf-4.tikz}}

\\
& \labeleq{B23}\; 
\InputIfFileExists{1-nf-5.tikz}{}{\input{./tikz/1-nf-5.tikz}}

\end{align*}
and we can conclude using Lemma~\ref{lem:co-copy-del} to rewrite the diagram in the dashed box into normal form.
\item For $\splitzero$, $\loseone$, and $\losezero$, the resulting diagrams are already in normal form:
$$

\InputIfFileExists{0-nf.tikz}{}{\input{./tikz/0-nf.tikz}}

$$ 
$$

\InputIfFileExists{lose0-nf.tikz}{}{\input{./tikz/lose0-nf.tikz}}
$$
$$ 

\InputIfFileExists{lose1-nf.tikz}{}{\input{./tikz/lose1-nf.tikz}}

$$
\item The cases of $\join$ and $\start$ can be dealt with using Lemma~\ref{lem:co-copy-del}.

\emph{i}) For $\join$. Since $d$ is in normal forl, we can find $d_1$ and $d_2$ such that
\begin{align*}

\InputIfFileExists{merge-nf.tikz}{}{\input{./tikz/merge-nf.tikz}}
\;&:=\;
\InputIfFileExists{merge-nf-1.tikz}{}{\input{./tikz/merge-nf-1.tikz}}

\\
& \labeleq{Lemma~\ref{lem:co-copy-del}}\quad 
\InputIfFileExists{merge-nf-2.tikz}{}{\input{./tikz/merge-nf-2.tikz}}

\end{align*}
with the last diagram in normal form.

\emph{ii}) For $\start$, similarly we have $d_1$ and $d_2$ such that
\begin{align*}

\InputIfFileExists{start-nf.tikz}{}{\input{./tikz/start-nf.tikz}}
\;&:=\;
\InputIfFileExists{start-nf-1.tikz}{}{\input{./tikz/start-nf-1.tikz}}

\\
& \labeleq{Lemma~\ref{lem:co-copy-del}}\quad 
\InputIfFileExists{start-nf-2.tikz}{}{\input{./tikz/start-nf-2.tikz}}

\\
& \labeleq{B3}\quad 
\InputIfFileExists{start-nf-3.tikz}{}{\input{./tikz/start-nf-3.tikz}}

\end{align*}
where $d_2$ is in normal form, by assumption.
\end{enumerate}
\end{proof}

The following lemma implies that, without loss of generality, we need only consider to left-to-right diagrams for our completeness proof.
\begin{lemma}\label{lem:compact}
There are natural bijections between the sets $\Syntax(v_1\!\objl\! v_2,w)$ and $\Syntax(v_1v_2,w\!\objr)$, and between $\Syntax(v,w_1\!\objl\! w_2)$ and $\Syntax(v\!\objr,w_1w_2)$, \emph{i.e.} between sets of string diagrams of the form
$$

\InputIfFileExists{wrong-way-left.tikz}{}{\input{./tikz/wrong-way-left.tikz}}
\quad\text{and}\quad 
\InputIfFileExists{right-way-right.tikz}{}{\input{./tikz/right-way-right.tikz}}

$$
$$
 
\InputIfFileExists{wrong-way-right.tikz}{}{\input{./tikz/wrong-way-right.tikz}}
\quad\text{and}\quad 
\InputIfFileExists{right-way-left.tikz}{}{\input{./tikz/right-way-left.tikz}}
 
$$
where $v,w, v_i, w_i$ are words over $\{\objr,\objl\}$.
\end{lemma}
\begin{proof}
The lemma holds in any compact closed category and relies on the ability to bend wires using $
\InputIfFileExists{cap-down.tikz}{}{\input{./tikz/cap-down.tikz}}
$ and $
\InputIfFileExists{cup-down.tikz}{}{\input{./tikz/cup-down.tikz}}
$. Explicitly, given a diagram of the first form, we can obtain one of the second form as follows:
\begin{equation}

\InputIfFileExists{wrong-way-left.tikz}{}{\input{./tikz/wrong-way-left.tikz}}
\quad \mapsto\quad 
\InputIfFileExists{bent-wires.tikz}{}{\input{./tikz/bent-wires.tikz}}

\end{equation}
The inverse mapping is given by the same wiring with the opposite direction. That they are inverse transformations follows immediately from the defining axioms of compact closed categories (A1-A2). The other bijection is constructed analogously. 
\end{proof}
\noindent Intuitively, Lemma~\ref{lem:compact} tells us that we can always bend incoming wires to the left and outgoing wires to the right before applying some equations, and recover the original orientation of the wires by bending them into their original place later.

\end{document}

%% file: macros.tex
\newcommand{\SMCeq}{\equiv}
\newcommand\PG[1][G]{\mathcal{#1}}
\newcommand{\OPGcat}{\mathsf{OPG}_\maxparity}
\newcommand{\oPGeq}{\simeq}
\newcommand{\IntC}{\mathsf{Int}}
\newcommand{\Play}{\mathsf{play}}
\newcommand{\Winning}{\mathcal{W}_\maxparity}
\newcommand{\possem}[1]{\{\!|#1|\!\}}
\newcommand{\oPGsem}[1]{\Winning\left(#1\right)}
\newcommand{\playsem}[1]{\{\!|#1|\!\}}
\newcommand{\sem}[1]{\left\llbracket{#1}\right\rrbracket}
\newcommand{\semexpr}[1]{\llparenthesis{\,#1\,}\rrparenthesis}
\newcommand\oPG[1]{\mathcal{#1}}

\newcommand{\ord}[1]{\left[#1\right]}
\newcommand\Labelling[1][\Omega]{#1}
\newcommand{\from}{\,:\,}
\newcommand{\Priorities}{\mathbb{P}}
\newcommand{\N}{\mathbb{N}}
\newcommand{\objr}{\blacktriangleright}
\newcommand{\objl}{\blacktriangleleft}
\newcommand{\idright}{
\InputIfFileExists{id.tikz}{}{\input{./tikz/id.tikz}}
}
\newcommand{\idleft}{
\InputIfFileExists{id-op.tikz}{}{\input{./tikz/id-op.tikz}}
}
\newcommand{\priority}[2]{\underline{#1}\left(#2\right)}
\newcommand{\choosezero}[2]{{#1}\lor{#2}}
\newcommand{\chooseone}[2]{{#1}\land{#2}}
\newcommand{\fix}[2]{\mu{#1}.{#2}}

\newcommand{\maxparity}{M}

\newcommand{\PRel}{\mathbf{PRel}}
\newcommand{\Semantics}{\OPGcat}
\newcommand{\Syntax}{\mathsf{Syn}_\maxparity}

\newcommand{\object}[1]{\overline{#1}}
\newcommand{\lin}[1]{\overrightarrow{#1}}
\newcommand{\rout}[1]{\overrightarrow{#1}}
\newcommand{\lout}[1]{\overleftarrow{#1}}
\newcommand{\rin}[1]{\overleftarrow{#1}}
\newcommand{\entries}[1]{\mathsf{in}^{#1}}
\newcommand{\exits}[1]{\mathsf{out}^{#1}}
\newcommand{\labeleq}[1]{\mathrel{\overset{\makebox[0pt]{\mbox{\normalfont\tiny\sffamily #1}}}{=}}}

\newcommand{\eqPG}{\labeleq{PGA}}

\newcommand{\Tr}{\mathsf{Tr}}
\newcommand{\entry}[1]{\mathtt{l}(#1)}
\newcommand{\exit}[1]{\mathtt{r}(#1)}
\newcommand{\inl}{\mathtt{l}}
\newcommand{\inr}{\mathtt{r}}
\newcommand{\Powerset}[1]{\mathcal{P}\left(#1\right)}
\newcommand{\puis}{\,;\,}
\newcommand{\id}{\mathsf{id}}
\newcommand{\Rel}{\mathbf{Rel}}

\usetikzlibrary{shapes.geometric,shapes.misc,matrix,arrows.meta,decorations.markings}
\pgfdeclarelayer{edgelayer}
\pgfdeclarelayer{nodelayer}
\pgfsetlayers{edgelayer,nodelayer,main}
\tikzset{x=1em, y=1em, baseline=-0.5ex}

\newcommand{\splitzero}{
\InputIfFileExists{player-0.tikz}{}{\input{./tikz/player-0.tikz}}
}
\newcommand{\splitone}{
\InputIfFileExists{player-1.tikz}{}{\input{./tikz/player-1.tikz}}
}
\newcommand{\join}{
\InputIfFileExists{merge.tikz}{}{\input{./tikz/merge.tikz}}
}
\newcommand{\start}{
\InputIfFileExists{start.tikz}{}{\input{./tikz/start.tikz}}
}
\newcommand{\loseone}{
\InputIfFileExists{lose-1.tikz}{}{\input{./tikz/lose-1.tikz}}
}
\newcommand{\losezero}{
\InputIfFileExists{lose-0.tikz}{}{\input{./tikz/lose-0.tikz}}
}
\newcommand{\priorityedge}[1]{
\begin{tikzpicture}
	\begin{pgfonlayer}{nodelayer}
		\node [style=priority] (37) at (0.75, 0) {$#1$};
		\node [style=none] (43) at (1.75, 0) {};
		\node [style=none] (44) at (2, 0) {};
		\node [style=none] (45) at (-0.5, 0) {};
		\node [style=none] (46) at (0, 0) {};
	\end{pgfonlayer}
	\begin{pgfonlayer}{edgelayer}
		\draw [style=arrow] (37) to (43.center);
		\draw (43.center) to (44.center);
		\draw [style=arrow] (45.center) to (46.center);
		\draw (46.center) to (37);
	\end{pgfonlayer}
\end{tikzpicture}
}

\newcommand{\diagbox}[3]{
\begin{tikzpicture}
	\begin{pgfonlayer}{nodelayer}
		\node [style=basic box] (0) at (0, 0) {$#1$};
		\node [style=none] (1) at (1.5, 0) {};
		\node [style=none] (2) at (-1.5, 0) {};
		\node [style=none] (3) at (1.5, 0.5) {\scriptsize $#3$};
		\node [style=none] (4) at (-1.5, 0.5) {\scriptsize $#2$};
	\end{pgfonlayer}
	\begin{pgfonlayer}{edgelayer}
		\draw (2.center) to (0);
		\draw (0) to (1.center);
	\end{pgfonlayer}
\end{tikzpicture}
}

\newcommand{\lrdiag}[3]{
\begin{tikzpicture}
	\begin{pgfonlayer}{nodelayer}
		\node [style=none] (0) at (-2.25, 0) {};
		\node [style=none] (1) at (-1.25, 0) {};
		\node [style=none] (2) at (-0.5, 0) {};
		\node [style=none] (3) at (2, 0.5) {\scriptsize $#3$};
		\node [style=basic box] (4) at (0, 0) {$#1$};
		\node [style=none] (5) at (1.5, 0) {};
		\node [style=none] (6) at (0.5, 0) {};
		\node [style=none] (7) at (2.25, 0) {};
		\node [style=none] (8) at (-2, 0.5) {\scriptsize $#2$};
	\end{pgfonlayer}
	\begin{pgfonlayer}{edgelayer}
		\draw [style=arrow] (0.center) to (1.center);
		\draw (1.center) to (2.center);
		\draw [style=arrow] (6.center) to (5.center);
		\draw (5.center) to (7.center);
	\end{pgfonlayer}
\end{tikzpicture}
}

%% file: tikz/id.tikz
\begin{tikzpicture}
	\begin{pgfonlayer}{nodelayer}
		\node [style=none] (36) at (0, 0) {};
		\node [style=none] (39) at (-1.25, 0) {};
		\node [style=none] (44) at (1, 0) {};
	\end{pgfonlayer}
	\begin{pgfonlayer}{edgelayer}
		\draw [style=arrow] (39.center) to (36.center);
		\draw (36.center) to (44.center);
	\end{pgfonlayer}
\end{tikzpicture}

%% file: tikz/id-op.tikz
\begin{tikzpicture}
	\begin{pgfonlayer}{nodelayer}
		\node [style=none] (36) at (-0.25, 0) {};
		\node [style=none] (39) at (1, 0) {};
		\node [style=none] (44) at (-1.25, 0) {};
	\end{pgfonlayer}
	\begin{pgfonlayer}{edgelayer}
		\draw [style=arrow] (39.center) to (36.center);
		\draw (36.center) to (44.center);
	\end{pgfonlayer}
\end{tikzpicture}

%% file: tikz/start.tikz
\begin{tikzpicture}
	\begin{pgfonlayer}{nodelayer}
		\node [style=co-copy] (37) at (0.75, 0) {};
		\node [style=none] (43) at (1.5, 0) {};
		\node [style=none] (44) at (2, 0) {};
	\end{pgfonlayer}
	\begin{pgfonlayer}{edgelayer}
		\draw [style=arrow] (37) to (43.center);
		\draw (43.center) to (44.center);
	\end{pgfonlayer}
\end{tikzpicture}

%% file: tikz/lose-1.tikz
\begin{tikzpicture}
	\begin{pgfonlayer}{nodelayer}
		\node [style=player 1] (0) at (-1.75, 0) {};
		\node [style=none] (1) at (-2.5, 0) {};
		\node [style=none] (4) at (-3.25, 0) {};
	\end{pgfonlayer}
	\begin{pgfonlayer}{edgelayer}
		\draw [style=arrow] (4.center) to (1.center);
		\draw (1.center) to (0);
	\end{pgfonlayer}
\end{tikzpicture}

%% file: tikz/lose-0.tikz
\begin{tikzpicture}
	\begin{pgfonlayer}{nodelayer}
		\node [style=player 0] (0) at (-1.75, 0) {};
		\node [style=none] (1) at (-2.5, 0) {};
		\node [style=none] (4) at (-3.25, 0) {};
	\end{pgfonlayer}
	\begin{pgfonlayer}{edgelayer}
		\draw [style=arrow] (4.center) to (1.center);
		\draw (1.center) to (0);
	\end{pgfonlayer}
\end{tikzpicture}

%% file: games.tikzstyles
\tikzstyle{none}=[inner sep=0pt]
\tikzstyle{co-copy}=[circle, draw=black, fill=black, inner sep=0pt, minimum size=4pt]
\tikzstyle{player 0}=[diamond, draw=black, fill=white, inner sep=.5pt, minimum size=8pt, node font={\scriptsize}]
\tikzstyle{player 1}=[rectangle, draw=black, fill=white, inner sep=1.5pt, minimum size=6pt, node font={\scriptsize}]
\tikzstyle{priority}=[draw, fill=white, rounded rectangle, rounded rectangle left arc=none, minimum height=1.2em, minimum width=1.4em, node font={\scriptsize}]
\tikzstyle{basic box}=[draw, fill=white, rectangle, minimum height=1.6em, minimum width=1.4em]
\tikzstyle{tall box}=[draw, fill=white, rectangle, minimum height=2em, minimum width=1.4em]
\tikzstyle{small box}=[draw, fill=white, rectangle, minimum height=1.2em, minimum width=1.4em, node font={\scriptsize}]
\tikzstyle{extra tall box}=[draw, fill=white, rectangle, minimum height=2.7em, minimum width=1.4em]

\tikzstyle{arrow}=[->]
\tikzstyle{colour arrow} = [->, blue]